\documentclass[draftclsnofoot,onecolumn,twoside,12pt]{IEEEtran}

\makeatletter
\def\ps@headings{%
\def\@oddhead{\mbox{}\scriptsize\rightmark \hfil \thepage}%
\def\@evenhead{\scriptsize\thepage \hfil \leftmark\mbox{}}%
\def\@oddfoot{}%
\def\@evenfoot{}}
\makeatother
\usepackage{amssymb,amsmath,color,graphicx}
\usepackage{subfigure}
\usepackage{cite}
\usepackage{verbatim}
\usepackage{algorithm}
\usepackage{algpseudocode}
\usepackage{multirow}
\allowdisplaybreaks

\newtheorem{theorem}{Theorem}[section]

\newtheorem{lemma}{Lemma}[section]

\newtheorem{proposition}{Proposition}[section]
\newtheorem{remark}[theorem]{Remark}

\ifodd 0
\newcommand{\com}[1]{\textbf{\color{red} (COMMENT: #1)}} 
\else

\newcommand{\com}[1]{}
\fi

\begin{document}

\title{Spatial Topology Adjustment for Minimizing Multicell Network Power Consumption}

\author{Taesoo Kwon,~\IEEEmembership{Member,~IEEE}

\thanks{T. Kwon is with the Electronics and Telecommunications Research Institute (ETRI), Daejeon, 305-700, South Korea (e-mail: tskwon@etri.re.kr).}
}

\markboth{Kwon: Spatial Topology Adjustment for Minimizing Multicell Network Power Consumption}
{Kwon: Spatial Topology Adjustment for Minimizing Multicell Network Power Consumption}

\maketitle

\begin{abstract}
  While the deployment of base stations (BSs) becomes increasingly dense in order to accommodate the growth in traffic demand, these BSs may be under-utilized during most hours except peak hours.
  Accordingly, the deactivation of these under-utilized BSs is regarded as the key to reducing network power consumption; however, the remaining active BSs should increase their transmit power in order to fill network coverage holes that result from BS switching off.
  This paper investigates the optimal balance between such beneficial and harmful effects of BS switching off in terms of minimizing the network power consumption, through comprehensively considering the spatial BS distribution, BS transmit power, BS power consumption behaviors, radio propagation environments, and frequency reuse.
  When BSs are deployed according to a homogeneous Poisson point process, the suboptimal and approximated design problems are formulated as geometric programming and the solutions lead to insightful design principles for the key design parameters including the spatial density, transmit power, and frequency reuse of remaining active BSs.
  The numerical results demonstrate that these solutions are very close to the optimal balances.

%
\end{abstract}


\section{Introduction} \label{sec:Intro}

  Recently, the wide dissemination of smart devices has accelerated the wireless traffic demand, and in order to accommodate this significant growth, the deployment of base stations (BSs) in wireless cellular networks continues to become more and more dense.
  This growth further aggravates concerns about the ever-increasing energy consumption and carbon footprint \cite{Bhargava11_GreenCellular}, when considering that the power consumption of BSs in cellular networks is 60--80\% of the operators' power consumption \cite{Fehske11_GlobalFootprint} and the energy provision for the BSs is to 50\% of the total operational cost \cite{Auer10_EnergyAwareRan}.
  Therefore, energy savings associated with the BSs are a crucial issue for decreasing not only harmful greenhouse gas emissions but also operational expenditures (OPEX).
  The efforts to reduce BS energy consumption are being undertaken at various levels \cite{Bhargava11_GreenCellular,Auer10_EnergyAwareRan}: the energy source level (e.g., adopting renewable energy resources \cite{McGuire11_RenewableEnergyBs}), component level (e.g., improving the power amplifier efficiency \cite{Brubaker09_EfficiencyPAs}), link and protocol level (e.g., discontinuous transmission for long standby BSs \cite{Frenger11_BsDtx}), and network topology level (e.g., heterogeneous network deployments with optimal balances of macro-, micro-, pico-, femto-cells, and relay stations \cite{Son11_EnergyHierarchicalCell,Fettweis11_GreenRelay}).

  In wireless cellular networks, BS deployment is typically designed to accommodate peak time traffic (e.g., evening hours) that is up to ten times higher than that of the off-peak periods (e.g., late night hours) and their traffic volume varies temporally in a comparatively regular pattern throughout the day \cite{Alu08_Traffic,Auer10_EnergyAwareRan}.
  Accordingly, it is likely that BSs are under-utilized during most hours except peak hours.
  However, a BS has poor energy efficiency particularly in low load situations; even at zero load, the direct current (DC) power consumption of a BS remains at approximately 50\% of the peak power \cite{Ferling10_EnergyEfficiency}.
  In this regard, the entire or partial deactivation of under-utilized BSs could be a key to saving energy \cite{Bhargava11_GreenCellular,Mancuso11_RedCostCellular,Auer11_EnergyWirelessNet}.
  However, the BS switching off technology creates a coverage deficit issue due to the deactivated BSs \cite{Oh11_DynamicEnergyCell,Han13_GreenCellCoop}.
  In order to fill this coverage deficit, the service area of the remaining active BSs must be expanded.
  Hence, the BS deactivation should be applied carefully through comprehensively considering both the energy saving effects of under-utilized BSs and the increased burden on the remaining BSs that is required in order to maintain service quality.
  This paper investigates the minimization of the total power consumption of BSs in multicell networks through quantifying and balancing a tradeoff between the beneficial and harmful effects of BS switching off.
  The optimal solution of the problems will enable the derivation of the inherent design principles that depict the relationship among the BS power consumption model, radio environments, and service constraints.


\subsection{Related Work} \label{ssec:RelatedWork}


  Many studies have proposed various algorithms for reducing access network power consumption by adapting the number of active BSs to the dynamic traffic demands and, more specifically, for determining when and which BSs should be activated and how the remaining active BSs should expand the physical coverage and accommodate the current network load in order to maintain service quality, e.g., \cite{Chiaraviglio08_EnergyUmts,Samdanis10_EnergySonCell,Samdanis10_EnergyDynCell,Marsan11_EnergyCoopCell,OhSon13_DynBsOnOff}.
  The suggestions in these studies have primarily focused on the design of specific algorithms and they have verified that the algorithms perform well for given BS power consumption models and radio environments regarding the minimization of the power consumption of the total BSs in a network via simulations.
  However, these studies have not provided comprehensive and insightful understandings about BS switching off because the performance of such specific algorithms depends heavily on the BS power consumption properties (e.g., how the transmit power contributes to the total power consumption and at which level the BSs should be deactivated) and the radio environments (e.g., the radio propagation model).

%
%
%


  Several studies have investigated partial or entire BS switching off based on the practical BS power consumption behavior.
  In general, the power consumption behavior of an active BS can be expressed as the sum of the transmit power-dependent or load-dependent part due to the radio transmission and the constant part remaining due to cooling, power supply, and monitoring \cite{Holtkamp13_BsPowerModel}.
  The work of \cite{Ferling10_EnergyEfficiency} introduced the discontinuous transmission (DTX) of a BS, which was motivated by an affine BS power consumption model, and it evaluated the performance based on the third generation partnership project long-term evolution (3GPP LTE) system parameters.
  The work of \cite{Peng11_PowerSavingCell} considered a practical BS power consumption model and addressed the effect of this model on the algorithm design and performance.
  In contrast, in \cite{Bhaumik10_BreatheCool}, it was claimed that when adjusting the cell size of the remaining active BSs, the optimal cell size depended significantly on the amount of fixed power consumption, i.e., a high fixed power consumption in the active BSs resulted in a reduction in the number of active BSs and an increasing in the cell size in terms of BS energy saving.
  Recently, the work of \cite{Holtkamp13_MinBsPwc} introduced the joint optimization of the BS transmission power reduction, DTX, and resource allocation through considering an affine BS power consumption model.
  However, these studies only rely on extensive simulations for the performance assessment of their algorithms.



  Due to the difficulty of expressing the effect of intercell interference in wireless multicell networks, the performance of BS switching off algorithms has primarily been evaluated through simulation.
  The work presented in \cite{Han12_EnergyCoopCell} mathematically analyzed the outage probability of the user equipment (UE) at the worst-case location, in order to quantify the effect of the cooperative coverage extension of the remaining active BSs.
  In that study, the worst-case location was considered to exclude the impact of the intercell interference, but the intercell interference is a key factor that system designers should not overlook.
  The studies presented in \cite{Shamai10_GreenUplink,Cao12_OptBsDensity} mathematically analyzed the effect of base station deactivation through statistically modeling the location of BSs as a homogeneous Poisson point process (PPP).
  The study in \cite{Shamai10_GreenUplink} analyzed the adverse effect of BS switching off on the uplink power consumption but only considered the intercell interference for a one-dimensional cell deployment.
  The study in \cite{Cao12_OptBsDensity} examined the optimal BS density under a service outage constraint in interference-limited homogeneous and heterogeneous cellular networks.
  However, neither of these studies have comprehensively considered the multicell operational parameters such as the BS switching off, BS power consumption model, BS transmit power, and frequency reuse.
\subsection{Contributions and Organization} \label{ssec:Contributions}

  This paper analytically investigates the optimal balance between the power saving that results from switched-off BSs and the load increase of the remaining active BSs in terms of minimizing the multicell network power consumption. 
  The key contributions are highlighted in the following three aspects.

  \subsubsection{Comprehensively formulating a spatial topology adjustment problem}
  Optimization problems are formulated with the objective of minimizing the network power consumption under network coverage and capacity constraints.
  These design problems are characterized by
    (i) the network power consumption model incorporating the affine power consumption model of an individual BS and
    (ii) the network coverage and capacity constraints being expressed as functions of the spatial density ($\lambda$), transmit power ($p$), and frequency reuse ($\beta$) of the remaining active BSs after switching off some BSs.
  These formulations provide baselines from which to build suboptimal or approximated design problems that enable closed form solutions and useful optimality conditions.

  \subsubsection{Deriving the design principles for green multicell planning}
  The suboptimal $\lambda$ and $p$ are derived in closed forms through replacing the network coverage constraint with its lower bound (LB) that enable the design problem to be convexified.
  This solution leads to the design principles for greening multicell networks, where the operation for the BS switching off and transmit power adjustment can be split into four cases according to the BS sleep mode effect, i.e., the difference ($\bar{P}$) between the standby power consumption of an active BS and the power consumption of a sleep mode BS.
  It is noteworthy that the interval of $\bar{P}$ and the values of $\lambda$ and $p$ for each case are expressed in terms of the spatial BS distribution, BS transmit power, affine BS power consumption model, required service quality, and radio propagation model.

  \subsubsection{Investigating the impact of frequency reuse on green multicell networks}
  The partial frequency reuse (PFR) that allows each BS to use only part of the spectrum is one of the most common methods for intercell interference mitigation.
  In order to investigate the effect of the PFR in terms of the network power consumption, the design problem to jointly optimize $\beta$ in addition to $\lambda$ and $p$ is formulated, where $\beta$ denotes the number of frequency bands for the PRF.
  The design problem is reexpressed as geometric programming (GP) and the necessary conditions for the optimality are derived.
  These conditions result in insightful design principles, e.g., $\frac{\lambda}{\beta}$ that indicates the spatial density of interfering BSs, i.e., the active BS density using the same subband, remains constant.

  The remainder of this paper is organized as follows.
  Section~\ref{sec:ModelProb} describes the system model for the multicell deployment and BS switching operation, and it defines the spatial performance metrics for the design problem formulation.
  Sections~\ref{sec:GreenCellFullFreq} formulates and solves the design problems for optimizing the BS switching off and transmit power adjustment in terms of minimizing the multicell network power consumption. Section~\ref{sec:GreenCellPartialFreq} optimizes the frequency reuse parameter in addition to the BS switching off and transmit power adjustments.
  Then, Section~\ref{sec:Results} discusses the numerical results, and Section~\ref{sec:Conclusions} concludes the paper.


\section{System Model and Performance Metric} \label{sec:ModelProb}

  This section describes the random multicell deployment and BS switching off models that render their performances mathematically tractable, and it defines the network performance metric used to incorporate the power consumption behaviors of an individual BS.


\subsection{BS Switching Off and Network Coverage Reduction} \label{subsec:GreenCellModel}


%
\begin{figure}[t]
\centering
\includegraphics[width=10cm]{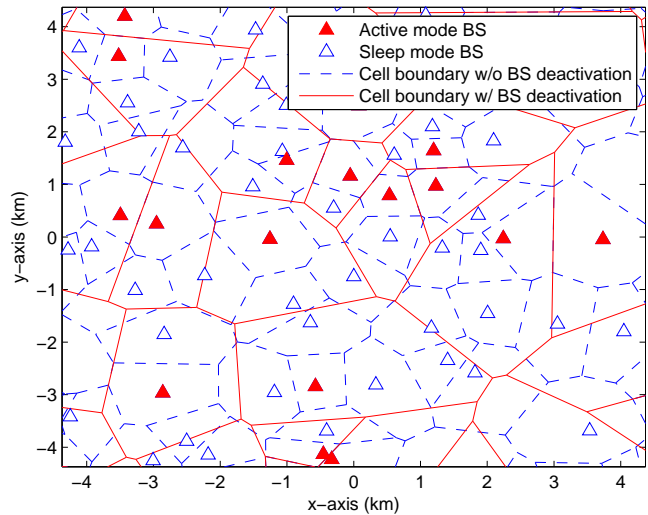}
\vspace{-0.5cm}
\caption{BS switching off and coverage in random multicell networks ($\lambda_u = 1\,\mathrm{km^{-2}}$, $\rho = 0.25$).}
\label{fig:GreenMulticellModel}
\end{figure}
%


  Each BS is in either an \emph{active} or a \emph{sleep} mode; the former means that the BS serves users within its coverage while the latter signifies that the BS does not transmit or receive signals.
  One simple method to select deactivated BSs is a random selection that determines their mode in a probabilistic manner.
  Let $\rho$ ($0 < \rho \leq 1$) denote the probability that a BS remains active.
  Each BS enters the sleep mode with probability $1-\rho$ while staying in the active mode with probability $\rho$.
  In order to avoid the coverage deficits that result from sleep mode BSs, the remaining active BSs fill in these coverage deficits through expanding their cell size.
  There are several methods to expand cell size, e.g., cooperative beamforming and transmit power increasing \cite{Peng11_PowerSavingCell, Han12_EnergyCoopCell}.
  This paper considers the scenario where the remaining active BSs expands their service coverage through increasing their transmit power $p$; thus, the power saving from the deactivated BSs and the transmit power increase of active BSs should be carefully balanced, i.e., $\rho$ that determines the BS mode must be appropriately traded off against $p$ in order to minimize the power consumption of the entire BSs while maintaining service quality above a specific level.


  In addition to the random selection of switched-off BSs, this paper considers the random BS deployment model, where BSs are randomly distributed according to a homogeneous PPP with a density of $\lambda_u$.
  This random BS model provides the cumulative distribution function (cdf) of signal to interference plus noise ratio (SINR) that is not only mathematically tractable but also tracks the performance of an actual BS deployment as accurately as the grid model \cite{Andrews11_Cellular}.
  Accordingly, it is expected that the performance of green multicell networks that randomly select switched-off BSs among randomly deployed BSs would be analytically tractable, which is the primary reason why this paper adopts these random models.
  More sophisticated methods, e.g., network adaptation accurately tracking spatial distribution of network loads \cite{Chiaraviglio08_EnergyUmts,Samdanis10_EnergySonCell,Samdanis10_EnergyDynCell,Marsan11_EnergyCoopCell,OhSon13_DynBsOnOff}, provide better performance, but they are difficult to use to quantitatively reveal the relationship among a variety of inherent properties of multicell networks, such as the BS power consumption behavior, radio propagation environment, and network-wide power consumption of BSs.
  Fig.~\ref{fig:GreenMulticellModel} presents the coverage models of a cellular network with full density $\lambda_u$ and with reduced density $\rho\lambda_u$ in random cellular networks.
  This figure depicts that the remaining active BSs expand their coverage due to the switching off of neighbor BSs.
  It is assumed that each user is served by the nearest active BS.


  In order to assess the coverage served by individual BSs more quantitatively, this paper addresses the coverage probability \cite{Andrews11_Cellular}, which is defined as the probability that a user experiences an SINR above a specified level, and this probability can be used as the key metric for expressing the network coverage constraint in minimizing the network-wide energy consumption of BSs.
  Let $\Xi$ denote a random variable that represents the SINR at which a user receives a downlink signal from the nearest BS.
  The coverage probability can be expressed as $\Pr\{ \Xi > \xi \}$ where $\xi$ is a given value that determines the minimum link quality.
  In order to address the network coverage more rigorously, the uplink signal quality as well as the downlink signal quality should be considered.
  However, in general, users select their serving BS through measuring the downlink signal quality from the BSs, e.g., the reference signal received power (RSRP), and the BS power consumption is dominated by the transmit power radiated from the BS antennas.
  For this reason, this paper only focuses on the downlink coverage probability in terms of network coverage constraints, and the consideration of uplink coverage remains as future work.
  The study in \cite{Andrews11_Cellular} derived a simple and tractable form expression for the downlink coverage probability in multicell networks using stochastic geometry, and this result encompasses the effect of various operational factors, such as the general path loss exponent, spatial density of BSs, and BS transmit power.
  Multicell networks that use the same frequency among cells suffer from low performance for cell boundary users due to severe intercell interference, e.g., $\Pr\{ \Xi > \xi \}$ is below 0.6 when $\xi = 0$dB \cite{Andrews11_Cellular}.
  A common method of resolving this intercell interference problem is to reduce the number of interfering BSs by allocating part of the entire frequency bands to each BS.
  Planned and dynamic frequency reuse methods contribute to further enhancing the performance of the cell boundary users, but their performance has primarily been evaluated through simulations.
  The study in \cite{Andrews11_Cellular} analytically quantified the performance enhancement by simply modeling this frequency reuse as a random frequency band allocation.
  This random frequency reuse provides a lower performance compared with more sophisticated methods, but its analytical performance results facilitate the mathematical formulation for designing various operations of multicell networks, e.g., \cite{Andrews11_FfrDownlink}.
  The following result provides the coverage probability in multicell networks where BSs are deployed according to a homogeneous PPP with $\lambda$ and they transmit their downlink signals with transmit power $p$ using one frequency band randomly chosen among $\beta$ frequency bands.
%
\begin{lemma}\label{lem:CoverProb}
  If one of $\beta$ frequency bands is randomly allocated to each cell, the coverage probability for the random multicell networks is given by
\begin{align}\begin{aligned}\label{eq:CoverPr}
  \psi(\lambda,p,\beta) \triangleq \Pr\{\Xi > \xi \} = \pi \lambda \int_0^\infty \exp\left( -\pi \lambda \left( 1 + \beta^{-1}\phi(\xi,\alpha) \right) x - \xi\nu^{-1} x^{\frac{\alpha}{2}} \right) dx.
\end{aligned}\end{align}
  where $\nu$ denotes the received signal to noise ratio (SNR) at unit distance, which depends on $p$ and $\beta$, and $\phi(\xi,\alpha) = \xi^{\frac{2}{\alpha}} \int_{\xi^{-\frac{2}{\alpha}}}^\infty \frac{1}{1+u^{\frac{\alpha}{2}}}\;du$.
  In particular, when $\alpha=4$, this coverage probability has the closed form expression, as follows:
\begin{align}\begin{aligned}\label{eq:CoverPrAlpha4}
  \psi(\lambda,p,\beta) = \frac{\pi^{\frac{3}{2}}\lambda \sqrt{\nu}}{\sqrt{\xi}} \exp\left( \frac{\left(\pi \lambda \left( 1+\beta^{-1}\phi(\xi,4) \right)\right)^2 \nu}{4\xi} \right) Q\left( \frac{\pi \lambda\sqrt{\nu} \left( 1+\beta^{-1}\phi(\xi,4) \right)}{\sqrt{2\xi}} \right),
\end{aligned}\end{align}
  where $\phi(\xi,4) = \sqrt{\xi} \left(\frac{\pi}{2} - \arctan\left(\frac{1}{\sqrt{\xi}}\right)\right)$ and $Q(x) \triangleq \frac{1}{\sqrt{2\pi}}\int_x^\infty \exp\left(-\frac{u^2}{2} \right) du$ denotes the $Q$-function.
\end{lemma}
\begin{proof}
  In order to derive \eqref{eq:CoverPr}, please refer to the proof of Theorem~4 in \cite{Andrews11_Cellular}.
  \eqref{eq:CoverPrAlpha4} can also be derived by using the integration formula of $\int_0^\infty \exp\left(-ax\right) \exp\left( -bx^2 \right) dx = \sqrt{\frac{\pi}{b}} \exp\left( \frac{a^2}{4b} \right) Q\left( \frac{a}{\sqrt{2b}} \right)$, similar to the one of equation (7) in \cite{Andrews11_Cellular}.
\end{proof}
\subsection{Individual BS Power Consumption and Multicell Area Power Consumption} \label{subsec:AreaPwModel}

  A BS consists of various components for power amplifying, air conditioning, signal processing and power supply; each component has an individual power consumption behavior that may be dependent on the transmit power or bandwidth \cite{Auer10_EnergyAwareRan}.
  Several studies have built a parameterized BS power consumption model that is expressed as an affine function of the BS transmit power \cite{Auer11_EnergyWirelessNet,Holtkamp13_BsPowerModel,Auer11_CellEnergyEva}, and the work of \cite{Holtkamp13_BsPowerModel} verified that this model approximates its underlying complex model well.
  In this model, it is noteworthy that the BS power consumption in the active mode almost linearly increases with the transmit power radiated at the BS antenna and a certain level of power is consumed even in the sleep mode.
  In contrast, it takes time to reactivate the BSs that are in the sleep mode, and this delay is dependent on which functional components are deactivated in the sleep mode.
  In this sense, the power consumption in the sleep mode, which is denoted by $P_s$, depends on at which level this mode is designed under the consideration of tradeoff between the amount of power saving and the reactivation delay of the components, e.g., switching off the entire BS or only core components such as the power amplifier and RF.
  Herein, $P_s=0$ indicates an ideal BS switching off.
  Let $P_c(X_j)$ denote the instantaneous power consumption of BS $j$ located at $X_j$.
  This paper assumes that each BS has a single sector and a single RF chain.
  Therefore, $P_c(X_j)$ can be expressed as follows:
\begin{align}\label{eq:PsModel}
  P_{c}(X_j) =
  \begin{cases}
    P_a + \Delta_p p, & \mbox{if BS $j$ is in the active mode, i.e., $0 < p \leq P_{max}$}\\
    P_s, & \mbox{if BS $j$ is in the sleep mode, i.e., $p=0$}
  \end{cases},
\end{align}
  where $P_a$ denotes the standby power consumption of the active BS that results from the components with transmit power independent power consumption characteristics, e.g., air conditioning and power supply, and $P_{max}$ denotes the maximum BS transmit power.
  Note that $P_a > P_s$ as a result of the sleep mode benefit.



  In order to assess the power consumption of the entire access network, this paper considers the average power consumption of BSs per unit area, i.e., the \emph{area power consumption} (APC), as the primary performance metric for capturing the energy efficiency of multicell networks.
  When BSs are distributed according to a homogeneous PPP with density $\lambda_u$ and they enter sleep mode with probability $1-\rho$, the spatial distribution of the remaining active BSs can be modeled as the independent thinning of homogeneous PPP $\Phi$ with retention probability $\rho$ \cite{Stoyan96_StochasticGeom}.
  That is, original point process $\Phi$ can be considered as the superposition of two independent homogeneous PPPs for active and sleep mode BSs, which are denoted by $\Phi_a$ with density $\rho \lambda_u$ and $\Phi_s$ with density $(1-\rho) \lambda_u$, respectively.
  Therefore, the APC is given as follows:
\begin{align}\begin{aligned}\label{eq:DefAreaPwCons}
  & \mathbb{E}\left\{\frac{\mbox{Network power consumption}}{\mbox{Network area}}\right\} = \mathbb{E}\left\{ \frac{\sum_{X_j \in \Phi_a} P_c (X_j) + \sum_{X_j \in \Phi_s} P_c (X_j)}{\int_{X \in \mathbb{R}^2} dX} \right\} \\
  & \stackrel{\mathrm{(a)}}{=} \frac{1}{\int_{X \in \mathbb{R}^2} dX} \left( \rho\lambda_u \int_{X \in \mathbb{R}^2} \left( P_a + \Delta_p p \right)  dX  + (1-\rho)\lambda_u \int_{X \in \mathbb{R}^2} P_s dX \right) \\
  & = \lambda_u \left( \rho(P_a + \Delta_p p) + (1-\rho)P_s  \right) 
  \stackrel{\mathrm{(b)}}{=} \lambda \left( \bar{P} + \Delta_p p \right) + \lambda_u P_s,
\end{aligned}\end{align}
%
  where $\mathbb{E}\left\{Y\right\}$ denotes the expectation of random variable $Y$, (a) follows from the Campbell theorem, the stationarity of a homogeneous PPP \cite{Stoyan96_StochasticGeom}, and \eqref{eq:PsModel}, and (b) follows from the introduction of new variables, i.e., $\lambda \triangleq \lambda_u \rho$ and $\bar{P} \triangleq P_a-P_s$.
  It is worth noting that the APC in \eqref{eq:DefAreaPwCons} naturally connects the parameterized BS power consumption model given in \eqref{eq:PsModel} with the network-wide BS power consumption.
  The next sections quantitatively reveal how the BS power consumption property affects the minimization of this APC.



\section{Area Power Consumption Minimization} \label{sec:GreenCellFullFreq}

  This section formulates and solves the problem for the balance between the spatial density of active BSs and their transmit power, i.e., $\lambda$ and $p$, in terms of minimizing the APC defined in \eqref{eq:DefAreaPwCons}.
  To begin with, in order to clarify the relationship between optimal $\lambda$ and $p$, this section only considers the universal frequency reuse (UFR) case, i.e., $\beta = 1$.
  The case of $\beta > 1$ is addressed in Section~\ref{sec:GreenCellPartialFreq}.

  In formulating the problem, the definition of the service quality constraint remains ambiguous.
  This paper addresses the service quality through discriminating between two factors: \emph{network coverage and capacity}.
\begin{itemize}
  \item \textbf{Network coverage constraint}:
    The network should be able to provide users with the spatial coverage with probability not less than $\eta$, i.e., $\Pr\left\{ \Xi > \xi \right\} \geq \eta$.
    In random multicell networks, the complementary cumulative distribution function (ccdf) of SINR, i.e., the coverage probability, is given by \eqref{eq:CoverPr}, which is a function of the spatial density of the active BSs and received SNR at a unit distance.
    The received SNR at the unit distance denoted by $\nu$ depends on the BS transmit power per physical resource block (PRB). 
    Let $q$ denote the BS transmit power per PRB.
    This paper assumes that $q$ is the same for all PRBs.
    In addition, the problem concentrates on the worst-case coverage scenario where all active BSs use all available PRBs.
    Under this scenario, the total transmit power of a BS is expressed as $p = qB$, where $B$ denotes the total number of available PRBs\footnote{Practical systems require much more complex transmit power calculation due to the transmit power boosting for control channels and reference signals and the transmit power control for link adaptation. However, this paper only considers this simple transmit power model for analytical tractability.}.
    This transmit power model pursues the reduction of the \emph{possible} APC rather than an exact APC. 
    When considering that it is almost impossible to accurately estimate the instantaneous transmit power of individual BSs, it is sensible to use this simple model that facilitates the derivation of analytically tractable results.
    When using this transmit power model, the received SNR at unit distance, i.e., $\nu$ becomes $\frac{Aq}{\bar{\sigma}^2} = \frac{Ap/B}{\bar{\sigma}^2} = \frac{p}{\sigma^2}$ where $\bar{\sigma}^2$ denotes the noise power per PRB, $A$ represents the path loss at a unit distance, and $\sigma^2 \triangleq \frac{B \bar{\sigma}^2}{A}$.
    Therefore, when $\beta=1$, the coverage probability given by \eqref{eq:CoverPr}, i.e., $\psi(\lambda, p, 1)$, can be reexpressed as follows:
\begin{align}\begin{aligned}\label{eq:CoverPrPsi}
  \psi(\lambda, p, 1) \triangleq \pi \lambda \int_0^\infty \exp\left( -\pi \lambda \left( 1 + \phi(\xi,\alpha) \right) x - \xi p^{-1} \sigma^2 x^{\frac{\alpha}{2}} \right) dx.
\end{aligned}\end{align}
%

  \item \textbf{Network capacity constraint}:
    The aggregate network throughput in a wireless cellular system almost linearly increases with the BS density \cite{Andrews12_SG_KtierHetNet}.
    That is, the spatial density of BSs is a dominant factor in determining the network capacity.
    Hence, it is sensible to roughly express the network capacity required for a certain period in terms of the spatial density of BSs that can accommodate a network load.
    Herein, there are two more justifications for the consideration of the BS density as capacity requirement during off-peak period.
    The first justification is the difficulty of obtaining an accurate estimate of the amount of required network capacity.
    The network load depends significantly on the variations in the number of users, user locations, and traffic patterns  \cite{Alu08_Traffic,Auer10_EnergyAwareRan}; thus, only a rough estimation is available. 
    The second justification is that operators typically deploy sufficient BSs to accommodate the peak time traffic and the BS switching off for energy saving is primarily applied during low loads, i.e., off peak times.
    This implies that the green multicell networks focused on in this paper are coverage limited rather than capacity limited; thus, this rough expression of the network capacity is sufficient to reflect the capacity requirement in low load situations.
    Based on these justifications, this paper imposes a capacity constraint expressed as $\lambda \geq \lambda_l$ in addition to the coverage constraint.
    This network capacity constraint corresponds to $\lambda_l \mathbb{E}\{\log_2\left(1+\Xi\right)\}$ in terms of the average area spectral efficiency (ASE), or $\lambda_l \log_2(1+\xi)$ in terms of the ASE associated with the per-cell spectral efficiency that can be spatially guaranteed with probability $\eta$, when considering the above network coverage constraint\footnote{The data rate of state-of-the art communication systems, e.g., the 3GPP LTE system, can be approximated simply as an attenuated form of the Shannon bound, i.e, $a\log_2(1+\Xi)$ bps/Hz where $a$ is an attenuation factor smaller than one \cite{LTE_LinkPerformance}. This paper considers an ideal link of $a=1$.}.
\end{itemize}
%



  The APC minimization problem under these two constraints can be formulated as follows:
\begin{subequations}\label{eq:MinArPw}
\begin{align}
    \underset{\lambda,\,p}{\text{minimize}} \quad   & \lambda(\bar{P}+\Delta_p p) \label{eq:MinArPw_Obj} \\
    \text{subject to} \quad & \psi(\lambda, p, 1) \geq \eta \label{eq:MinArPw_CnstCoverage} \\
    & \lambda_l \leq \lambda \leq \lambda_u \label{eq:MinArPw_CnstLambda} \\
    & 0 < p \leq P_{max} \label{eq:MinArPw_CnstTxPw}.
\end{align}
\end{subequations}
  Note that problem \eqref{eq:MinArPw} minimizes the objective function given by \eqref{eq:MinArPw_Obj} because $\lambda_u P_s$ in \eqref{eq:DefAreaPwCons} is constant.
  The left hand side (lhs) of constraint \eqref{eq:MinArPw_CnstCoverage}, which is given in \eqref{eq:CoverPrPsi}, includes the integration and this term makes it difficult to solve problem \eqref{eq:MinArPw} and express its optimal solution as a simple form.
  The following remark provides useful insight about the above optimization problem.
\begin{remark}\label{rmk:Increasing}
  Assume that problem \eqref{eq:MinArPw} is feasible.
  Then,

  (i) if $(\lambda_1,p_1)$ and $(\lambda_2,p_2)$ belong to the feasible set of problem \eqref{eq:MinArPw} and $(\lambda_1,p_1) \preceq (\lambda_2,p_2)$\footnote{Herein, $(x_1,y_1) \preceq (x_2,y_2)$ denotes that $x_1 \leq x_2$ and $y_1 \leq y_2$.}, $\psi(\lambda_1,p_1,1) \leq \psi(\lambda_2,p_2,1)$;

  (ii) if $(\lambda^*,p^*)$ denotes the optimal solution of problem \eqref{eq:MinArPw}, $\psi(\lambda^*,p^*,1) = \eta$.
\end{remark}
\begin{proof}
  (i) The increase in the coverage probability with $\lambda$ and $p$ is obvious because the effect of the noise power deceases as $\lambda$ and $p$ increase.
  For analytical proof, see Appendix~\ref{app:proof:rmk:Increasing}.

  (ii) Assume that $\psi(\lambda^*,p^*,1) > \eta$.
  From the increasing property of $\psi(\lambda,p,1)$ with respect to $\lambda$ and $p$, there is $\tilde{p}$ such that $0< \tilde{p} < p^* < P_{max}$ and $\psi(\lambda^*,\tilde{p},1) = \eta$; thus, $(\lambda^*,\tilde{p},1)$ is also feasible.
  The APC for $(\lambda^*,\tilde{p})$ is smaller than that of $(\lambda^*,p^*)$, because objective function \eqref{eq:MinArPw_Obj} decreases with $p$ for given $\lambda$.
  This contradicts that $(\lambda^*,p^*)$ is the optimal solution. 
\end{proof}
  Define $p^*(\lambda,\beta)$ and $\lambda^*(p,\beta)$\footnote{This section only considers the case of $\beta=1$ but the definitions of $p^*(\lambda,\beta)$ and $\lambda^*(p,\beta)$ can be extended into those of $\beta \geq 1$. In these definitions, $\psi(\lambda,p,\beta)$ in \eqref{eq:CoverPr} is more definitely expressed in \eqref{eq:CoverPrPsiPfr}.} as
\begin{align}
    p^*(\lambda,\beta) \triangleq \mbox{$p$ such that } \psi(\lambda,p,\beta) = \eta \mbox{ for given $\lambda$ and $\beta$}, \label{eq:Pgiven} \\
    \lambda^*(p,\beta) \triangleq \mbox{$\lambda$ such that } \psi(\lambda,p,\beta) = \eta \mbox{ for given $p$ and $\beta$}. \label{eq:Lambdagiven}
\end{align}
%
  Then, $p^*(\lambda,1)$ (resp. $\lambda^*(p,1)$) can be easily found via the bisection method because $\psi(\lambda,p,1)$ is a monotonically increasing function with respect to $p$ for fixed $\lambda$ (resp. $\lambda$ for fixed $p$).
  Remark~\ref{rmk:Increasing} implies that when one variable is given, the other variable can be readily optimized; that is, for given $\lambda = \bar{\lambda}$ where $\lambda_l \leq \bar{\lambda} \leq \lambda_u$ (resp. for given $p = \bar{p}$ where $0 < \bar{p} \leq P_{max}$), the optimal $p$ (resp. optimal $\lambda$) becomes $p^*(\bar{\lambda},1)$ if $0 < p^*(\bar{\lambda},1) \leq P_{max}$ (resp. $\lambda^*(\bar{p},1)$ if $\lambda_l \leq \lambda^*(\bar{p},1) \leq \lambda_u$).


  However, the original problem in \eqref{eq:MinArPw} should be jointly optimized for $\lambda$ and $p$, and it is not easy to analytically express the solution of this joint optimization problem.
  Note that even when $\alpha=4$, which has a closed form expression of $\psi(\lambda, p, 1)$ given by \eqref{eq:CoverPrAlpha4}, the analytical expression for the solution of problem \eqref{eq:MinArPw} is not easily derived.
  The next subsection reformulates the nonconvex problem in \eqref{eq:MinArPw} as GP through recasting coverage constraint \eqref{eq:MinArPw_CnstCoverage} and it derives the optimal active BS density and transmit power of the reformulated problem.


\subsection{Formulation and Solution of Tightened Design Problem} \label{subsec:DesignUfr}

  The coverage probability in \eqref{eq:CoverPrPsi}, which is an integral function, needs to be expressed more elegantly in order to find a good balance between $\lambda$ and $p$.
  From the property that $\exp(-x) \geq 1 - x$ for $x \geq 0$\footnote{Let $f(x) \triangleq \exp(-x)+x-1$ for $x \geq 0$. Because $\frac{df(x)}{dx} \geq 0$ for $x \geq 0$ and $f(0)=0$, $f(x) \geq 0$ when $x \geq 0$.},
\begin{align}\begin{aligned}\label{eq:CoverPrPsiLb}
  \psi(\lambda, p, 1) & \geq \pi \lambda \int_0^\infty \exp\left( -\pi \lambda \left( 1 + \phi(\xi,\alpha) \right) x \right) \left( 1 - \xi p^{-1} \sigma^2 x^{\frac{\alpha}{2}} \right) dx \\
    & = \frac{1}{1+\phi(\xi,\alpha)} \left( 1 - \frac{\xi \sigma^2}{\left( \pi\lambda \left( 1+\phi(\xi,\alpha) \right) \right)^{\frac{\alpha}{2}} p} \Gamma\left( 1+\frac{\alpha}{2} \right) \right).
\end{aligned}\end{align}
  where $\textstyle \Gamma(x)\triangleq\int_{0}^{\infty} t^{x-1}\exp(-t)dt$ denotes the Gamma function.
  The design problem typically aims to maintain high coverage probability, i.e., $\eta$ as close to one as possible, even when some BSs are switched off.
  In this regard, the design interest is placed on a small noise scenario where the noise effect is limited, i.e., $\sigma^2 \left( \lambda^{\frac{\alpha}{2}} p \right)^{-1}$ is small for given $\xi$.
  It is worth noting that the lower bound of $\psi(\lambda, p, 1)$ in \eqref{eq:CoverPrPsiLb} gets tighter as $\sigma^2 \left( \lambda^{\frac{\alpha}{2}} p \right)^{-1}$ decreases and the equality holds when $\sigma^2 = 0$.
  In fact, $\psi(\lambda, p, 1)$ for $\sigma^2 > 0$ is upper bounded by the one for $\sigma^2 = 0$; thus, for $\sigma^2 > 0$, $\psi(\lambda, p, 1) < \frac{1}{1+\phi(\xi,\alpha)}$.
  In this sense, it is expected that the LB of \eqref{eq:CoverPrPsiLb} replaces $\psi(\lambda, p, 1)$ well.
  Note that \eqref{eq:CoverPrPsiLb} is similar to equation (9) in \cite{Andrews11_Cellular}.
  However, therein, the coverage probability was approximated for a small but nonzero noise power by applying $\exp(-x) = 1 - x + o(x)$ where $o(x)$ denotes the higher order terms of $x$, rather than being lower bounded.
  From \eqref{eq:CoverPrPsiLb}, the new coverage constraint that tightens \eqref{eq:MinArPw_CnstCoverage} is given as follows:
\begin{align}\label{eq:CnstCoverageLb}
    \lambda^{\frac{\alpha}{2}} p \geq \theta(\xi,\eta,\alpha) \triangleq \frac{\xi \sigma^2 \Gamma\left( 1+\frac{\alpha}{2} \right)}{\left( 1-\eta \left( 1+\phi(\xi,\alpha) \right) \right) \left( \pi \left( 1+\phi(\xi,\alpha) \right) \right)^{\frac{\alpha}{2}}}.
\end{align}
  The above condition follows from \eqref{eq:CoverPrPsiLb} and therefore condition \eqref{eq:CnstCoverageLb} is a sufficient condition for meeting constraint \eqref{eq:MinArPw_CnstCoverage}.
  Because $\psi(\lambda, p, 1) < \frac{1}{1+\phi(\xi,\alpha)}$ as mentioned above, $\eta$ representing coverage requirement needs to be less than $\frac{1}{1+\phi(\xi,\alpha)}$ for the sake of feasibility.
  This means that $\theta(\xi,\eta,\alpha) > 0$.
  It is worth noting that constraint \eqref{eq:CnstCoverageLb} exhibits that the BS density has the equivalent effect as the transmit power to the power of $\frac{\alpha}{2}$ in terms of the coverage probability.
  In two-dimensional Euclidean space, scaling up the BS density by a factor of $s$ is equivalent to scaling down the distance between a UE and its serving BS averagely by a factor of $s^{\frac{1}{2}}$
  ; thus, this distance reduction corresponds to increasing transmit power $p$ to $s^{\frac{\alpha}{2}} p$ in terms of the coverage probability.
  In this sense, constraint \eqref{eq:CnstCoverageLb} makes sense.

  Now, by using this new condition instead of constraint \eqref{eq:MinArPw_CnstCoverage} in problem \eqref{eq:MinArPw}, a new tightened optimization problem is yielded, as follows:
\begin{subequations}\label{eq:MinArPwTight}
\begin{align}
    \underset{\lambda,\,p}{\text{minimize}} \quad   & \lambda(\bar{P}+\Delta_p p) \label{eq:MinArPwTight_Obj} \\
    \text{subject to} \quad & \lambda^{\frac{\alpha}{2}} p \geq \theta(\xi,\eta,\alpha) \label{eq:MinArPwTight_CnstCoverage} \\
    & \lambda_l \leq \lambda \leq \lambda_u \label{eq:MinArPwTight_CnstLambda} \\
    & 0 < p \leq P_{max} \label{eq:MinArPwTight_CnstTxPw}.
\end{align}
\end{subequations}
  In order to guarantee the feasibility of problem \eqref{eq:MinArPwTight}, it is assumed that sufficient condition \eqref{eq:CnstCoverageLb} is met when all the BSs transmit with the maximum transmit power, i.e., $\lambda_u^{\frac{\alpha}{2}} P_{max} > \theta(\xi,\eta,\alpha)$.
  Note that the optimal value of problem \eqref{eq:MinArPwTight} provides the upper bound of that of problem \eqref{eq:MinArPw}.
%
\begin{proposition}\label{prop:OptPwCnsmTight}
  Assume that $\lambda_u^{\frac{\alpha}{2}} P_{max} > \theta(\xi,\eta,\alpha)$ and $\eta < \frac{1}{1+\phi(\xi,\alpha)}$.
  When $(\hat{\lambda},\hat{p})$ denotes the optimal solution of problem \eqref{eq:MinArPwTight},
\begin{align}\label{eq:OptSolTight}
  (\hat{\lambda},\hat{p}) =
  \begin{cases}
    \left( \lambda_u, \theta(\xi,\eta,\alpha) \lambda_u^{-\frac{\alpha}{2}} \right), & \hspace{-3cm}\mbox{if $0 < \bar{P} < \Delta_p(\frac{\alpha}{2}-1)\theta(\xi,\eta,\alpha) \lambda_u^{-\frac{\alpha}{2}}$ } \\
    \left( \left( \Delta_p(\frac{\alpha}{2}-1)\theta(\xi,\eta,\alpha) \bar{P}^{-1} \right)^{\frac{2}{\alpha}}, \Delta_p^{-1} (\frac{\alpha}{2}-1)^{-1}\bar{P} \right), \\
        & \hspace{-8.8cm}\mbox{if $\Delta_p(\frac{\alpha}{2}-1)\theta(\xi,\eta,\alpha) \lambda_u^{-\frac{\alpha}{2}} \leq \bar{P} < \min\left\{ \Delta_p(\frac{\alpha}{2}-1)\theta(\xi,\eta,\alpha) \lambda_l^{-\frac{\alpha}{2}},  \Delta_p \left( \frac{\alpha}{2}-1 \right) P_{max} \right\}$} \\
    \left( \left(\theta(\xi,\eta,\alpha) P_{max}^{-1}\right)^{\frac{2}{\alpha}},P_{max} \right),
        & \hspace{-3cm}\mbox{if $\bar{P} \geq \Delta_p \left( \frac{\alpha}{2}-1 \right) P_{max} $ and $\lambda_l^{\frac{\alpha}{2}} P_{max} < \theta(\xi,\eta,\alpha)$} \\
    \left( \lambda_l, \theta(\xi,\eta,\alpha) \lambda_l^{-\frac{\alpha}{2}} \right), & \hspace{-4.3cm}\mbox{if $\bar{P} \geq \Delta_p(\frac{\alpha}{2}-1)\theta(\xi,\eta,\alpha) \lambda_l^{-\frac{\alpha}{2}}$ and $\lambda_l^{\frac{\alpha}{2}} P_{max} \geq \theta(\xi,\eta,\alpha)$}.
  \end{cases}
\end{align}
\end{proposition}
\begin{proof}
  See Appendix~\ref{app:proof:prop:OptPwCnsmTight}.
\end{proof}
  It is interesting that the operation for the spatial topology adjustment is split into four cases, according to the difference between the standby power consumption in the active mode and the power consumption in the sleep mode, i.e., $\bar{P} = P_a - P_s$ that expresses the amount of power consumption saved through the BS switching off.
  This supports that the operation depends on the effect of the BS sleep mode.
  The result of Proposition~\ref{prop:OptPwCnsmTight} can be interpreted as follows:
  \begin{itemize}
    \item[(i)] The case of $0 < \bar{P} < \Delta_p(\frac{\alpha}{2}-1)\theta(\xi,\eta,\alpha) \lambda_u^{-\frac{\alpha}{2}}$ indicates the environment where the effect of the sleep mode is not dominant and accordingly the network minimizes the power consumption through reducing the transmit power rather than through switching off BSs.
    \item[(ii)] The case of $\Delta_p(\frac{\alpha}{2}-1)\theta(\xi,\eta,\alpha) \lambda_u^{-\frac{\alpha}{2}} \leq \bar{P} < \min\left\{ \Delta_p(\frac{\alpha}{2}-1)\theta(\xi,\eta,\alpha) \lambda_l^{-\frac{\alpha}{2}},  \Delta_p \left( \frac{\alpha}{2}-1 \right) P_{max} \right\}$ denotes the instance that requires a tradeoff of the power reduction from the switched-off BSs against the increase in the transmit power of the remaining active BSs in a coverage limited scenario.
        The result reveals the interesting outcome that \emph{$\hat{p}$ is linearly proportional to $\bar{P}$ with a slope of $\frac{1}{\Delta_p \left( \frac{\alpha}{2} - 1 \right)}$} and $\hat{\lambda}$ is set to the solution of $\lambda^2 \hat{p} = \theta(\xi,\eta)$.
        This causes the setting of transmit power for balancing $\lambda$ and $p$ to be very handy.
    \item[(iii)] The case of $\bar{P} \geq \Delta_p \left( \frac{\alpha}{2}-1 \right) P_{max} $ and $\lambda_l^{\frac{\alpha}{2}} P_{max} < \theta(\xi,\eta,\alpha)$ represents the situation where the BS switching off enables a substantial energy saving effect due to a large standby consumption in the active mode, a very small sleep mode power consumption, or both.
        Hence, the density of the active BSs is maintained as small as possible.
        $\lambda_l^{\frac{\alpha}{2}} P_{max} < \theta(\xi,\eta,\alpha)$ denotes that the network is coverage limited when $\lambda$ is minimized; hence, this result determines $\lambda$ through maximizing $p$.
    \item[(iv)]The case of $\bar{P} \geq \Delta_p(\frac{\alpha}{2}-1)\theta(\xi,\eta,\alpha) \lambda_l^{-\frac{\alpha}{2}}$ and $\lambda_l^{\frac{\alpha}{2}} P_{max} \geq \theta(\xi,\eta,\alpha)$ designates that the network benefits from the BS switching off similar to that in (iii) but the active BSs are too sparse to accommodate the expected network traffic.
        Therefore, the density of the active BSs is at least equal to $\lambda_l$ and the transmit power should be set to meet the coverage constraint.
  \end{itemize}
  In summary, the first case only reduces the BS transmit power, the third and fourth cases minimize the BS density, and the second case appropriately balances the BS density and transmit power.

  The solution in \eqref{eq:OptSolTight} also reveals the effect of the radio propagation, i.e., $\alpha$, in addition to the BS power consumption behaviors.
  This $\alpha$ affects both the interval of $\bar{P}$ that determines the BS operation and the value of $(\hat{\lambda},\hat{p})$ for the APC minimization, and its effect depends on other environmental factors, e.g., the noise power and target SINR.

\subsection{Improvement of Suboptimal Solution} \label{subsec:AlgorithmUfr}

  Recall that constraint \eqref{eq:CnstCoverageLb} is the sufficient condition of original constraint \eqref{eq:MinArPw_CnstCoverage}.
  Therefore, $\psi(\hat{\lambda},\hat{p},1)$ based on $(\hat{\lambda},\hat{p})$ in \eqref{eq:OptSolTight} is at least larger than $\eta$ as long as $\sigma^2 > 0$.
  Because $\psi(\lambda,p,1)$ is monotonically increasing with respect to $p$ for given $\lambda$ from Remark~\ref{rmk:Increasing}, $p^*(\hat{\lambda},1) < \hat{p}$ where $p^*(\hat{\lambda},1)$ defined in \eqref{eq:Pgiven} can be readily obtained via the bisection method.
  Hence, $(\hat{\lambda},p^*(\hat{\lambda},1))$ is a feasible solution of original problem \eqref{eq:MinArPw} and provides the APC less than $(\hat{\lambda},\hat{p})$.
  That is, based on $(\hat{\lambda},\hat{p})$ in \eqref{eq:OptSolTight}, a better solution $(\hat{\lambda},p^*(\hat{\lambda},1))$ can be readily yielded.


\section{Frequency Reuse and Green Multicell Networks} \label{sec:GreenCellPartialFreq}

  The partial frequency reuse (PFR) that allows each BS to use only part of the spectrum is a common and useful method for mitigating intercell interference problems.
  When the entire spectrum is divided into $\beta$ frequency bands, one of which an individual BS uses, it is known that as $\beta$ increases, the outer-cell user performance is improved while the average rate of each cell may be reduced \cite{Andrews11_Cellular}.
  Meanwhile, the question about the impact of this PFR on the power consumption in multicell networks is an interesting topic.
  This section formulates and solves the spatial topology design problem for green multicell networks through considering the number of frequency bands in addition to the BS density and transmit power as design variables.

  In order to reflect the impact of $\beta \geq 1$\footnote{$\beta$ is a positive integer, but this paper allows $\beta$ to be real for analytical convenience.}, the coverage and capacity constraints addressed in Section~\ref{sec:GreenCellFullFreq} must be appropriately modified.
\begin{itemize}
  \item \textbf{Network coverage constraint}:
    The PFR does not only decrease the density of the downlink interferers by a factor of $\beta$, but it also enables the boosting of the SNR on one PRB by concentrating the BS transmit power on an available frequency band.
    This impact is described by the results in Lemma~\ref{lem:CoverProb}, and, herein, $\nu$ is set to $\frac{Ap/(B/\beta)}{\bar{\sigma}^2} = \frac{Ap\beta}{B\bar{\sigma}^2} = \frac{p\beta}{\sigma^2}$, in pursuit of the possible APC minimization justified in Section~\ref{sec:GreenCellFullFreq}.
    That is,
\begin{align}\begin{aligned}\label{eq:CoverPrPsiPfr}
  \psi(\lambda, p, \beta) = \pi \lambda \int_0^\infty \exp\left( -\pi \lambda \left( 1 + \beta^{-1}\phi(\xi,\alpha) \right) x - \xi p^{-1} \beta^{-1} \sigma^2 x^{\frac{\alpha}{2}} \right) dx.
\end{aligned}\end{align}
    $\psi(\lambda, p, \beta)$ for $\beta=1$ is equal to $\psi(\lambda, p, 1)$ in \eqref{eq:CoverPrPsi}; thus, $\psi(\lambda, p, \beta)$ encompasses the coverage probability for the UFR as well as the PFR.
  \item \textbf{Network capacity constraint}:
    For the ease of formulation, the capacity constraint is expressed in terms of the BS density, which is similar to that in Section~\ref{sec:GreenCellFullFreq}.
    Because the network capacity is proportional to the bandwidth as well as the BS density for a given SINR distribution, this constraint can be simply cast into $\frac{\lambda}{\beta} \geq \lambda_l$, where $\beta$ in the denominator reflects the reduction of the available bandwidth due to the PFR.
    However, this model tends to underestimate the average ASE given by $\frac{\lambda}{\beta} \mathbb{E}\{\log_2\left(1+\Xi\right)\}$, where the distribution of $\Xi$ is also improved as $\beta$ increases.
    In contrast, it abstracts the performance well in terms of the ASE associated with the per-cell spectral efficiency that can be spatially guaranteed with probability $\eta$, i.e, $\frac{\lambda}{\beta} \log_2(1+\xi)$ for specific constant $\eta$.
    Considering the difficulty of its accurate estimation and the bias to the coverage limited scenario, as described in Section~\ref{sec:GreenCellFullFreq}, this simple capacity constraint is imposed on the design problem for green cellular networks.
\end{itemize}
%

\subsection{Formulation and Solution of Tightened or Approximated Design Problems} \label{subsec:DesignPfr}

  In order to solve the design problem, the coverage constraint based on coverage probability \eqref{eq:CoverPrPsiPfr} needs to be recast into a simpler form.
  Similar to Section~\ref{sec:GreenCellFullFreq}, the coverage probability in \eqref{eq:CoverPrPsiPfr} can be lower bounded by using $\exp(-x) \geq 1 - x$, as follows:
\begin{align}\begin{aligned}\label{eq:CoverPrPsiE}
  \psi(\lambda, p, \beta) \geq \psi_e(\lambda, p, \beta) \triangleq \frac{1}{1+\beta^{-1}\phi(\xi,\alpha)} \left( 1 - \frac{\xi \sigma^2}{\left( \lambda\pi \left( 1+\beta^{-1}\phi(\xi,\alpha) \right) \right)^{\frac{\alpha}{2}} p \beta} \Gamma\left( 1+\frac{\alpha}{2} \right) \right)
\end{aligned}\end{align}
  However, \eqref{eq:CoverPrPsiE} does not provide the relationship among $\lambda$, $p$, and $\beta$ that facilitates the convexification of the design problem.
  For this reason, another LB or approximation of the coverage probability for PFR is considered based on the LB of \eqref{eq:CoverPrPsiE}, as follows:
\begin{align}\begin{aligned}\label{eq:CoverPrPsiBeta}
  \psi(\lambda, p, \beta) \geq \psi_\beta(\lambda, p, \beta) \triangleq \frac{1}{1+\beta^{-1}\phi(\xi,\alpha)} \left( 1 - \frac{\xi \sigma^2}{\left( \lambda\pi \right)^{\frac{\alpha}{2}} p \beta} \Gamma\left( 1+\frac{\alpha}{2} \right) \right),
\end{aligned}\end{align}
\begin{align}\begin{aligned}\label{eq:CoverPrPsiA}
  \psi(\lambda, p, \beta) \approx \psi_a(\lambda, p, \beta) \triangleq \frac{1}{1+\beta^{-1}\phi(\xi,\alpha)} \left( 1 - \frac{\xi \sigma^2}{\left( \lambda\pi \left( 1+\phi(\xi,\alpha) \right) \right)^{\frac{\alpha}{2}} p \beta} \Gamma\left( 1+\frac{\alpha}{2} \right) \right).
\end{aligned}\end{align}

  In contrast, only when $\alpha=4$, the LB that enables to convexify the design problem can also be obtained from $Q(x) > \frac{1}{\sqrt{2\pi}} \left( \frac{x}{1+x^2} \right) \exp \left( -\frac{x^2}{2} \right)$, as follows:
\begin{align}\begin{aligned}\label{eq:CoverPrPsiQ}
  \psi(\lambda,p,\beta) & = \frac{\pi^{\frac{3}{2}}\lambda \sqrt{p \beta}}{\sqrt{\xi \sigma^2}} \exp\left( \frac{\left(\pi \lambda \left( 1+\beta^{-1}\phi(\xi,4) \right)\right)^2 p \beta}{4\xi \sigma^2} \right) Q\left( \frac{\pi \lambda\sqrt{p \beta} \left( 1+\beta^{-1}\phi(\xi,4) \right)}{\sqrt{2\xi \sigma^2}} \right) \\
  & > \psi_q(\lambda, p, \beta)
    \triangleq \frac{1}{1 + \beta^{-1}\phi(\xi,4)} \left( 1 - \frac{ 2\xi \sigma^2   }{2\xi \sigma^2 + \left( \lambda \pi \left( 1+\beta^{-1}\phi(\xi,4) \right) \right)^2 p \beta} \right).
\end{aligned}\end{align}
  Note that $\psi(\lambda,p,\beta) < \frac{1}{1 + \beta^{-1}\phi(\xi,4)}$ from $Q(x) < \frac{1}{\sqrt{2\pi}} \left( \frac{1}{x} \right) \exp \left( -\frac{x^2}{2} \right)$, and this upper bound is equal to $\psi(\lambda,p,\beta)$ for $\sigma^2 = 0$.
  In addition, the LB in \eqref{eq:CoverPrPsiQ} becomes tighter and eventually approaches $\frac{1}{1 + \beta^{-1}\phi(\xi,4)}$ as $\frac{\lambda^2 p}{\sigma^2}$ increases. 
  In this regard, $\psi_q(\lambda,p,\beta)$ in \eqref{eq:CoverPrPsiQ} approximates the original coverage probability given by $\eqref{eq:CoverPrAlpha4}$ very well for $\eta$ close to one.
  Moreover, it can be shown that $\psi_q(\lambda,p,\beta) > \psi_e(\lambda,p,\beta)$ when $\alpha=4$, by comparing \eqref{eq:CoverPrPsiQ} with \eqref{eq:CoverPrPsiE}.

  In summary, the four methods for recasting the original coverage probability were introduced, and it is worth noting that the latter three equations, i.e., $\psi_\beta$, $\psi_a$, and $\psi_q$, enables to convexify the design problem.
  In addition, $\psi_\beta$ and $\psi_q$ are lower bounds for $\psi$ while $\psi_a$ is not.
  Their relations are summarized in Table~\ref{tbl:CoverProbRel}.
\begin{table}[t]
\caption{Relationship among functions associated to coverage probability.} \centering
\label{tbl:CoverProbRel}
    \begin{tabular}{|c|c|}
    \hline
        \textbf{Scenario} & \textbf{Relationship}
    \\ \hline \hline
        $\beta = 1$, $\alpha > 2$ & $\psi_\beta < \psi_e = \psi_a < \psi$
    \\ \hline
        $\beta = 1$, $\alpha = 4$ & $\psi_\beta < \psi_e = \psi_a < \psi_q < \psi$
    \\ \hline
        $\beta > 1$, $\alpha > 2$ & $\psi_\beta < \psi_e < \min\left\{ \psi_a,\,\psi \right\}$
    \\ \hline
        $\beta > 1$, $\alpha = 4$ & $\psi_\beta < \psi_e < \psi_q < \psi$
    \\ \hline
    \end{tabular}
\vspace{-0.5cm}
\end{table}
%

  New coverage constraints based on $\psi_\beta$, $\psi_a$, and $\psi_q$ for replacing original coverage probability \eqref{eq:CoverPrPsiPfr} can be reexpressed as the following common form:
\begin{align}\label{eq:CnstCoverageSuffFrPosy}
    c_0 \lambda^{-\frac{\alpha}{2}} p^{-1} \beta^{-1} + c_1 \beta^{-1} + c_2 \beta^{-2} \leq 1.
\end{align}
  where $c_0$, $c_1$, and $c_2$ are nonnegative constants and given as follows:
\begin{align}
    & \mbox{For $\psi_\beta$ with $\alpha>2$: } \,\,
        c_0 = \frac{\xi \sigma^2 \Gamma\left( 1+\frac{\alpha}{2} \right)}{\pi^{\frac{\alpha}{2}}(1 - \eta)}, \,\,
        c_1 = \frac{\eta}{1 - \eta}\phi(\xi,\alpha), \,\,
        c_2 = 0. \label{eq:PfrPosyCoeffBeta} \\
    & \mbox{For $\psi_a$ with $\alpha>2$: } \,\,
        c_0 = \frac{\xi \sigma^2 \Gamma\left( 1+\frac{\alpha}{2} \right)}{\pi^{\frac{\alpha}{2}}(1 - \eta)\left( 1+\phi(\xi,\alpha) \right)^{\frac{\alpha}{2}}}, \,\,
        c_1 = \frac{\eta}{1 - \eta}\phi(\xi,\alpha), \,\,
        c_2 = 0. \label{eq:PfrPosyCoeffA} \\
    & \mbox{For $\psi_q$ with $\alpha=4$: } \,\,
        c_0 = \frac{2\xi \sigma^2 \eta}{\pi^2(1 - \eta)}, \,\,
        c_1 = \frac{2\eta-1}{1 - \eta}\phi(\xi,4), \,\,
        c_2 = \frac{\eta}{1 - \eta}\phi^2(\xi,4). \label{eq:PfrPosyCoeffQ}
\end{align}
  It is sensible to assume that $\eta > \frac{1}{2}$, because $\eta$ close to one is considered in this design problem.
  This assumption guarantees that all of $c_0$, $c_1$, and $c_2$ in \eqref{eq:PfrPosyCoeffBeta}, \eqref{eq:PfrPosyCoeffA}, and \eqref{eq:PfrPosyCoeffQ} are positive.

  In problem \eqref{eq:MinArPw} for $\beta=1$, i.e., the UFR, by imposing new coverage constraint \eqref{eq:CnstCoverageSuffFrPosy} and capacity constraint $\beta \lambda_l \leq \lambda \leq \lambda_u$ instead of constraints \eqref{eq:MinArPw_CnstCoverage} and \eqref{eq:MinArPw_CnstLambda}, the design problem for the PFR, i.e., $\beta \geq 1$, can be formulated as follows:
\begin{subequations}\label{eq:MinArPwTightPfrGeom}
\begin{align}
    \underset{\lambda>0,\,p>0,\,\beta>0}{\text{minimize}} \quad   & \lambda(\bar{P}+\Delta_p p) \label{eq:MinArPwTightPfr_Obj} \\
    \text{subject to} \quad & c_0  \lambda^{-\frac{\alpha}{2}} p^{-1} \beta^{-1} + c_1 \beta^{-1} + c_2 \beta^{-2} \leq 1 \label{eq:MinArPwTightPfr_CnstCoverage} \\
    & \lambda_l \lambda^{-1}\beta \leq 1 \label{eq:MinArPwTightPfr_CnstLambdaL} \\
    & \lambda_u^{-1} \lambda \leq 1 \label{eq:MinArPwTightPfr_CnstLambdaU} \\
    & p P^{-1}_{max} \leq 1 \label{eq:MinArPwTightPfr_CnstTxPw} \\
    & \beta^{-1} \leq 1 \label{eq:MinArPwTightPfr_CnstDelta}.
\end{align}
\end{subequations}
  Note that objective function \eqref{eq:MinArPwTightPfr_Obj} and constraint \eqref{eq:MinArPwTightPfr_CnstCoverage} are posynomial functions.
  That is, problem \eqref{eq:MinArPwTightPfrGeom} is GP in posynomial form \cite{Boyd07_CvxOpt}.
  GP can be transformed to a convex problem, which can be solved by the interior-point method \cite{Boyd07_CvxOpt} or available solver \cite{Boyd12_CvxTool}.
  However, the solution based on these iterative methods does not reveal the analytical relationship among the design parameters such as $\lambda$, $p$, and $\beta$.
  In this regard, the necessary conditions for the optimal solution of problem \eqref{eq:MinArPwTightPfrGeom} are derived.

\begin{proposition}\label{prop:OptPwCnsmPfr}
  Assume that $\eta > \frac{1}{2}$, and let $( \breve{\lambda},\breve{p},\breve{\beta} )$ denote the optimal solution of problem \eqref{eq:MinArPwTightPfrGeom}.
  These optimal variables are related as follows:
\begin{align}\begin{aligned}\label{eq:OptPwCnsmPfr}
    \breve{\lambda} = \lambda_l \breve{\beta}, \hspace{1cm}
    f(\breve{\lambda},\breve{p},\breve{\beta}) = 1,
\end{aligned}\end{align}
  where $f(\lambda,p,\beta) \triangleq c_0  \lambda^{-\frac{\alpha}{2}} p^{-1} \beta^{-1} + c_1 \beta^{-1} + c_2 \beta^{-2}$.
  To be more specific, $( \breve{\lambda},\breve{p},\breve{\beta} )$ satisfies one of the following conditions:
\begin{itemize}
    \item[(i)] $\breve{\lambda} = \lambda_u$, $\breve{\beta} = \frac{\lambda_u}{\lambda_l}$, $\breve{p} = \frac{c_0 \lambda_u^{-\frac{\alpha}{2}} \left( \lambda_u / \lambda_l \right)}{\left( \lambda_u / \lambda_l \right)^2 - c_1 \left( \lambda_u / \lambda_l \right) - c_2}$.
    \item[(ii)] $\breve{\lambda} = \lambda_l \bar{\beta}$, $\breve{p} = \frac{c_0 \lambda_l^{-\frac{\alpha}{2}} \bar{\beta}^{-\frac{\alpha}{2}+1}}{\bar{\beta}^2 - c_1\bar{\beta} - c_2}$, $\breve{\beta} = \bar{\beta}$, where $\bar{\beta} \in \left\{ \beta | \, g(\beta) = 0,\,\beta > 1 \right\}$.
    \item[(iii)] $\breve{\lambda} = \lambda_l \bar{\beta}$, $\breve{p} = P_{max}$, $\breve{\beta} = \bar{\beta}$, where $\bar{\beta} \in \left\{ \beta | \, f(\lambda_l \beta,P_{max},\beta) = 1,\,\beta > 1 \right\}$.
    \item[(iv)] $\breve{\lambda} = \lambda_l$, $\breve{\beta} = 1$, $\breve{p} = \frac{c_0 \lambda_l^{-\frac{\alpha}{2}}}{1 - c_1 - c_2}$.
\end{itemize}
  In (ii), $g(\beta)$ is given by $g_\alpha(\beta)$ with coefficients \eqref{eq:PfrPosyCoeffBeta} for $\psi_\beta$ in \eqref{eq:CoverPrPsiBeta} (resp. \eqref{eq:PfrPosyCoeffA} for $\psi_a$ in \eqref{eq:CoverPrPsiA}) when $\alpha>2$, whereas it can also be given by $g_4(\beta)$ with coefficients \eqref{eq:PfrPosyCoeffQ} for $\psi_q$ in \eqref{eq:CoverPrPsiQ} only when $\alpha=4$, where
\begin{align}
    & g_\alpha(\beta) \triangleq \beta^{\frac{\alpha}{2}+2} - 2c_1 \beta^{\frac{\alpha}{2}+1} + c_1^2 \beta^{\frac{\alpha}{2}} - \frac{c_0 \Delta_p \lambda_l^{-\frac{\alpha}{2}} \bar{P}^{-1} \alpha}{2} \beta + c_0 c_1 \Delta_p \lambda_l^{-\frac{\alpha}{2}} \bar{P}^{-1} \left( \frac{\alpha}{2}-1 \right), \label{eq:BetaEqAlphaGen} \\
    & g_4(\beta) \triangleq \beta^4 - 2c_1 \beta^3 + \left( c_1^2 - 2c_2 \right) \beta^2 + 2 \left( c_1 c_2 - c_0 \Delta_p \lambda_l^{-2} \bar{P}^{-1} \right) \beta + \left( c_2^2 + c_0 c_1 \Delta_p \lambda_l^{-2} \bar{P}^{-1} \right). \label{eq:BetaEqAlpha4}
\end{align}
\end{proposition}
\begin{proof}
  See Appendix~\ref{app:proof:prop:OptPwCnsmPfr}.
\end{proof}
  The interpretation of the necessary conditions (i)--(iv) in Proposition~\ref{prop:OptPwCnsmPfr} corresponds to that of the four cases in \eqref{eq:OptSolTight} for the UFR, even though the results for the PFR do not explicitly offer the conditions that are expressed in terms of the interval of $\bar{P}$.
  One notable difference is regarding the additional variable $\beta$, and this variable is conditioned on $\breve{\lambda} = \lambda_l \breve{\beta}$, like \eqref{eq:OptPwCnsmPfr}.
  Note that the density of the interfering BSs is expressed as $\frac{\lambda}{\beta}$ in the PFR.
  Therefore, condition $\breve{\lambda} = \lambda_l \breve{\beta}$ signifies that \emph{the density of the interfering BSs is always maintained as a fixed value of $\lambda_l$} in terms of minimizing the APC in problem \eqref{eq:MinArPwTightPfrGeom}.
  Conditions (i), (iii), and (iv) in Proposition~\ref{prop:OptPwCnsmPfr} correspond to the extreme cases where one of constraints \eqref{eq:MinArPwTightPfr_CnstLambdaU}, \eqref{eq:MinArPwTightPfr_CnstTxPw}, and \eqref{eq:MinArPwTightPfr_CnstDelta} is active.
  In contrast, condition (ii) is relevant to the balance $\lambda$, $p$, and $\beta$, and the selection of these design variables begins with solving $g_\alpha(\beta) = 0$ in \eqref{eq:BetaEqAlphaGen} or $g_4(\beta) = 0$ in \eqref{eq:BetaEqAlpha4}.

\subsection{Derivation of Suboptimal Solutions} \label{subsec:AlgorithmPfr}

  The candidate solution for problem \eqref{eq:MinArPwTightPfrGeom} can be obtained through comparing the APC for $(\lambda,p,\beta)$ values derived from conditions (i)--(iv) in Proposition~\ref{prop:OptPwCnsmPfr} and then choosing the feasible one with the least APC among them.
  With a slight misuse of notation, this candidate solution is denoted as $(\breve{\lambda}, \breve{p}, \breve{\beta})$.
  Recall that problem \eqref{eq:MinArPwTightPfrGeom} uses the LB or approximation instead of the original coverage probability for its coverage constraint.
  Furthermore, it can be readily proven that the two properties in Remark~\ref{rmk:Increasing} also hold for $\psi(\lambda,p,\beta)$ in \eqref{eq:CoverPrPsiPfr} by replacing $\beta=1$ with given $\beta \geq 1$.
  In this sense, similar to the method presented in Section~\ref{subsec:AlgorithmUfr}, in order to return to the original coverage constraint, new solution $( \breve{\lambda},p^*(\breve{\lambda},\breve{\beta}),\breve{\beta} )$ is derived from candidate solution $(\breve{\lambda}, \breve{p}, \breve{\beta})$, where $p^*(\breve{\lambda},\breve{\beta})$ is defined in \eqref{eq:Pgiven}.
  For the design problem based on the LBs of $\psi$, i.e., $\psi_\beta$ and $\psi_q$, this new solution reduces the APC compared with that from $( \breve{\lambda},\breve{p},\breve{\beta} )$.
  In contrast, for the design problem based on the approximation of $\psi$, i.e., $\psi_a$, this method renders the solution feasible if $\psi(\breve{\lambda},\breve{p},\breve{\beta}) < \eta$ while decreasing the APC if $\psi(\breve{\lambda},\breve{p},\breve{\beta}) > \eta$.

\subsection{Design of $\lambda$ and $p$ for Given $\beta$} \label{subsec:DesignPfrGvnBeta}

  Unlike the previous subsections that considered $\beta$ as a design variable in addition to $\lambda$ and $p$, one can consider the optimization of $\lambda$ and $p$ for given $\beta \geq 1$.
  Given $\beta$, the design problem can be formulated by replacing $\psi(\lambda, p, 1) \geq \eta$ and $\lambda_l \leq \lambda \leq \lambda_u$ with $\psi(\lambda, p, \beta) \geq \eta$ and $\beta \lambda_l \leq \lambda \leq \lambda_u$, respectively, in problem \eqref{eq:MinArPw}.
  From \eqref{eq:CoverPrPsiE}, the coverage constraint in integral form, i.e., $\psi(\lambda, p, \beta) \geq \eta$ for given $\beta$, can be tightened by the following constraint in monomial form:
\begin{align}\label{eq:CnstCoverageLbGvnBeta}
    \lambda^{\frac{\alpha}{2}} p \geq \vartheta(\beta,\xi,\eta,\alpha) \triangleq \frac{\beta^{-1} \xi \sigma^2 \Gamma\left( 1+\frac{\alpha}{2} \right)}{\left( 1-\eta \left( 1+\beta^{-1}\phi(\xi,\alpha) \right) \right) \left( \pi \left( 1+\beta^{-1}\phi(\xi,\alpha) \right) \right)^{\frac{\alpha}{2}}}.
\end{align}
  As a result, the optimal solution of the tightened design problem for given $\beta$, i.e., $(\hat{\lambda}, \hat{p})$, can be obtained from Proposition~\ref{prop:OptPwCnsmTight} by replacing $\theta$ and $\lambda_l$ with $\vartheta$ and $\beta\lambda_l$, respectively.
  In addition, because $\psi(\lambda,p,\beta)$ in \eqref{eq:CoverPrPsiPfr} monotonically increases with respect to $p$ for given $\lambda$ and $\beta$ as addressed in Section~\ref{subsec:AlgorithmPfr}, $(\hat{\lambda}, p^*(\hat{\lambda},\beta))$ brings the APC less than that of $(\hat{\lambda}, \hat{p})$. 

\begin{table}[t]
\caption{System parameters.} \centering
\label{tbl:Parameters}
    \begin{tabular}{|c||c|c|}
    \hline
        \textbf{Variable} & \textbf{Description} & \textbf{Value}
    \\ \hline 
        $\lambda_u$ & Total density of deployed BSs & $1\,\mathrm{km}^{-2}$
    \\ 
        $\lambda_l$ & Minimum required BS density & $0.2\,\mathrm{km}^{-2}$
    \\ 
        $P_{max}$ & Maximum BS transmit power & $49\,\mathrm{dBm}$
    \\ 
        $P_{a}$ & Standby power consumption of an active BS & $185\,\mathrm{watts}$
    \\ 
        $\Delta_{p}$ & Slope of affine BS power consumption model & $4.7$
    \\ 
        $A$ & Path loss at unit distance $1\,\mathrm{km}$ & $-128.1\,\mathrm{dB}$
    \\ 
        $B\bar{\sigma}^2$ & Total noise power over $20\,\mathrm{MHz}$ bandwidth & $-100.99\,\mathrm{dBm}$
    \\ 
        $\sigma^2$ & Normalized noise power over total bandwidth ($B\bar{\sigma}^2/A$) & $-27.11\,\mathrm{dB}$
    \\ 
        $\alpha$ & Path loss exponent & 4 or 5
    \\ 
        $\xi$ & Minimum required SINR & $-6\,\mathrm{dB}$
    \\ \hline
    \end{tabular}
\vspace{-0.5cm}
\end{table}
%

\section{Numerical Results and Discussion} \label{sec:Results}

  This section evaluates the APC performance of the green multicell networks designed in Sections~\ref{sec:GreenCellFullFreq} and \ref{sec:GreenCellPartialFreq}, and it examines the effect of the affine BS power consumption model given by \eqref{eq:PsModel}.
  The system parameters for the performance evaluation are presented in Table~\ref{tbl:Parameters}, and these parameter values are used to obtain results in this section, unless stated otherwise.
  The parameter values for the BS power consumption and path loss models are taken from \cite{Holtkamp13_MinBsPwc,3GPP_TR36814}.

  Figs.~\ref{fig:ResultPrCvrgSinr} and \ref{fig:ResultContour} examine the coverage probability for the design problems proposed in Sections~\ref{sec:GreenCellFullFreq} and \ref{sec:GreenCellPartialFreq}, prior to evaluating the APC performance of multicell networks.

\begin{figure*}[!t]
\centerline{
    \subfigure[BS density: $\lambda = 0.2\mathrm{km}^{-2}$]{\includegraphics[width=9cm]{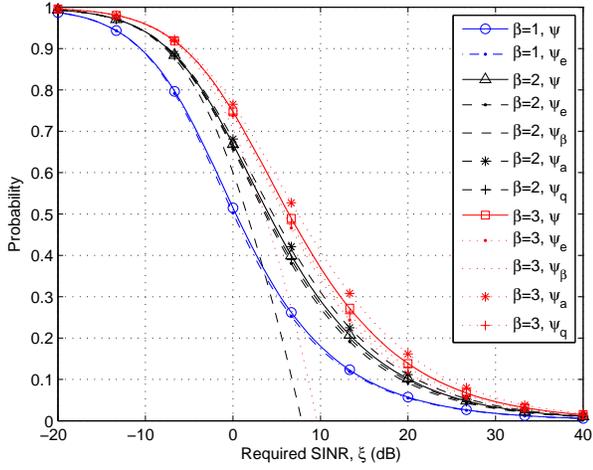}
    \label{fig:ResultPrCvrgSinrLowDensity}}
\hfil
    \subfigure[BS density: $\lambda = 0.5\mathrm{km}^{-2}$]{\includegraphics[width=9cm]{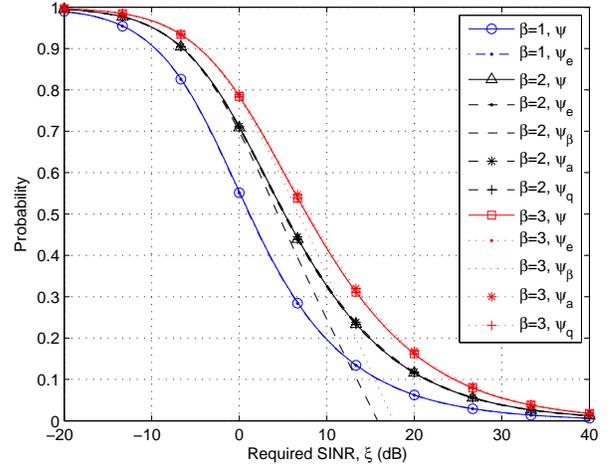}
    \label{ResultPrCvrgSinrHighDensity}}}
    \caption{Coverage probability vs. required SINR ($\alpha=4$, $p=39\mathrm{dBm}$; in symbolled-solid lines for $\psi$, symbols are yielded from analytical results while solid lines are collected from experimental results).}
\label{fig:ResultPrCvrgSinr}
\vspace{-0.5cm}
\end{figure*}
  Fig.~\ref{fig:ResultPrCvrgSinr} compares the coverage probability and its alternatives when 20\% and 50\% of the total BSs are functioning.
  In this evaluation, each BS sets its transmit power to be 10\% of $P_{max}$.
  For the exact coverage probability, i.e., $\psi$, it is observed that the analysis results coincide precisely with the empirical results; thus, the correctness of all analytical results is validated.
  This figure also demonstrates that the coverage probability is quite well approximated using its alternatives, which become closer for higher BS densities.
  This results from the noise power becoming increasingly overwhelmed by the aggregate interference as the BS density increases, and $\psi_e$, $\psi_{\beta}$, $\psi_q$, and $\psi_a$ provide better lower bounds and a better approximation for smaller $\sigma^2$ values.
  More specifically, for $\beta>1$, $\psi_e$, $\psi_q$, and $\psi_a$ tightly approximate $\psi$ for the entire range of coverage probability, while $\psi_{\beta}$ only approximates well for high coverage probabilities.
  However, when considering that systems typically support high coverage probabilities, e.g., above 0.8 or 0.9, $\psi_{\beta}$ also remains a tight LB of $\psi$.

\begin{figure*}[!t]
\centerline{
    \subfigure[Path loss exponent: $\alpha = 4$]{\includegraphics[width=9cm]{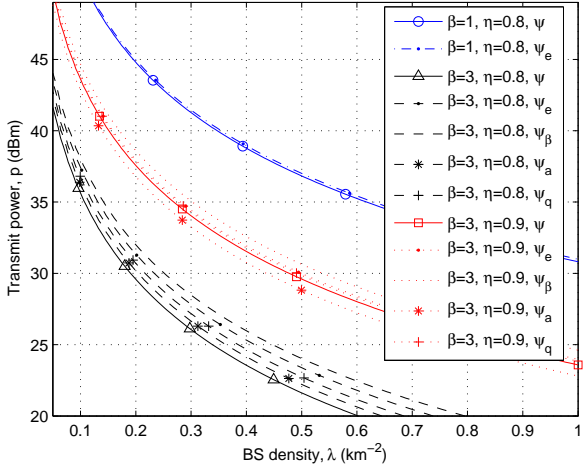}
    \label{fig:ResultContourAlpha4}}
\hfil
    \subfigure[Path loss exponent: $\alpha = 5$]{\includegraphics[width=9cm]{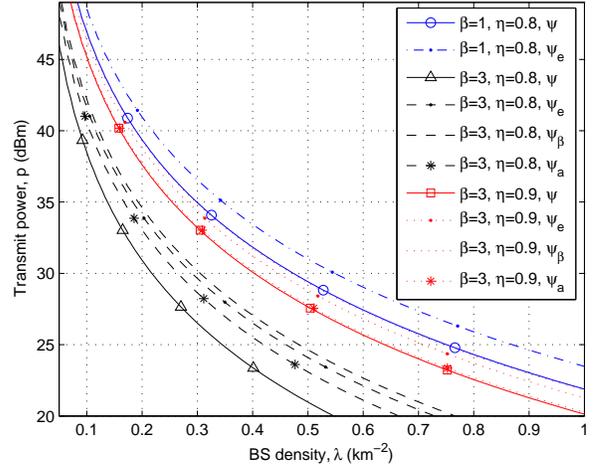}
    \label{fig:ResultContourAlpha5}}}
    \caption{Contour graph for coverage constraints.}
\label{fig:ResultContour}
\vspace{-0.5cm}
\end{figure*}

  Fig.~\ref{fig:ResultContour} presents the contour graphs of the coverage constraints based on $\psi$, $\psi_e$, $\psi_{\beta}$, $\psi_q$, and $\psi_a$, for high coverage probabilities.
  Unlike Fig.~\ref{fig:ResultPrCvrgSinr} for fixed $\lambda$ and $p$, these contour graphs enable to more precisely examine the relationship between $\lambda$ and $p$ and compare $\psi$ and its alternatives by focusing on a specific $\eta$.
  As summarized in Table~\ref{tbl:CoverProbRel}, it is observed that $\psi$ is lower bounded more and more tightly in the order of $\psi_\beta$, $\psi_e$, and $\psi_q$, whereas $\psi_a \geq \eta$ may tighten or relax $\psi \geq \eta$ depending on $\eta$ when $\alpha=4$.
  Recall that when $\beta > 1$, $\psi_q$ does not only enable the convexification of the coverage constraint but also provides tighter LB than that of $\psi_e$, but it is only available for $\alpha=4$.
  The increase in $\alpha$ causes a quicker decay in both the desired and interference signals, and accordingly the impact of the noise power becomes more significant.
  For this reason, the LBs of $\psi$ in Fig.~\ref{fig:ResultContourAlpha5} tend to have larger deviations than those in Fig.~\ref{fig:ResultContourAlpha4}.


\begin{figure*}[!t]
\centerline{
    \subfigure[Path loss exponent: $\alpha = 4$]{\includegraphics[width=9cm]{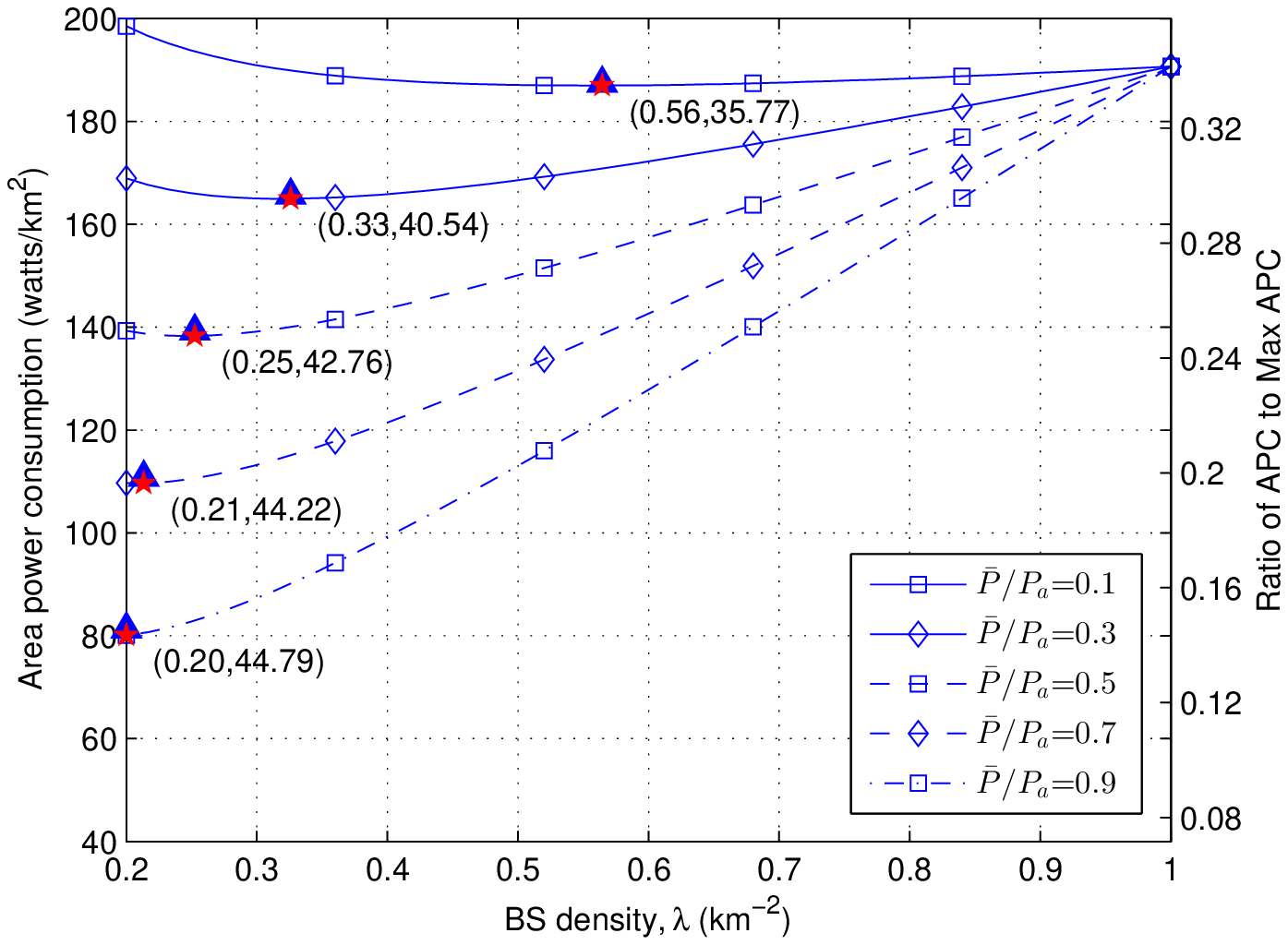}
    \label{fig:ResultApcMinUfrAlpha4}}
\hfil
    \subfigure[Path loss exponent: $\alpha = 5$]{\includegraphics[width=9cm]{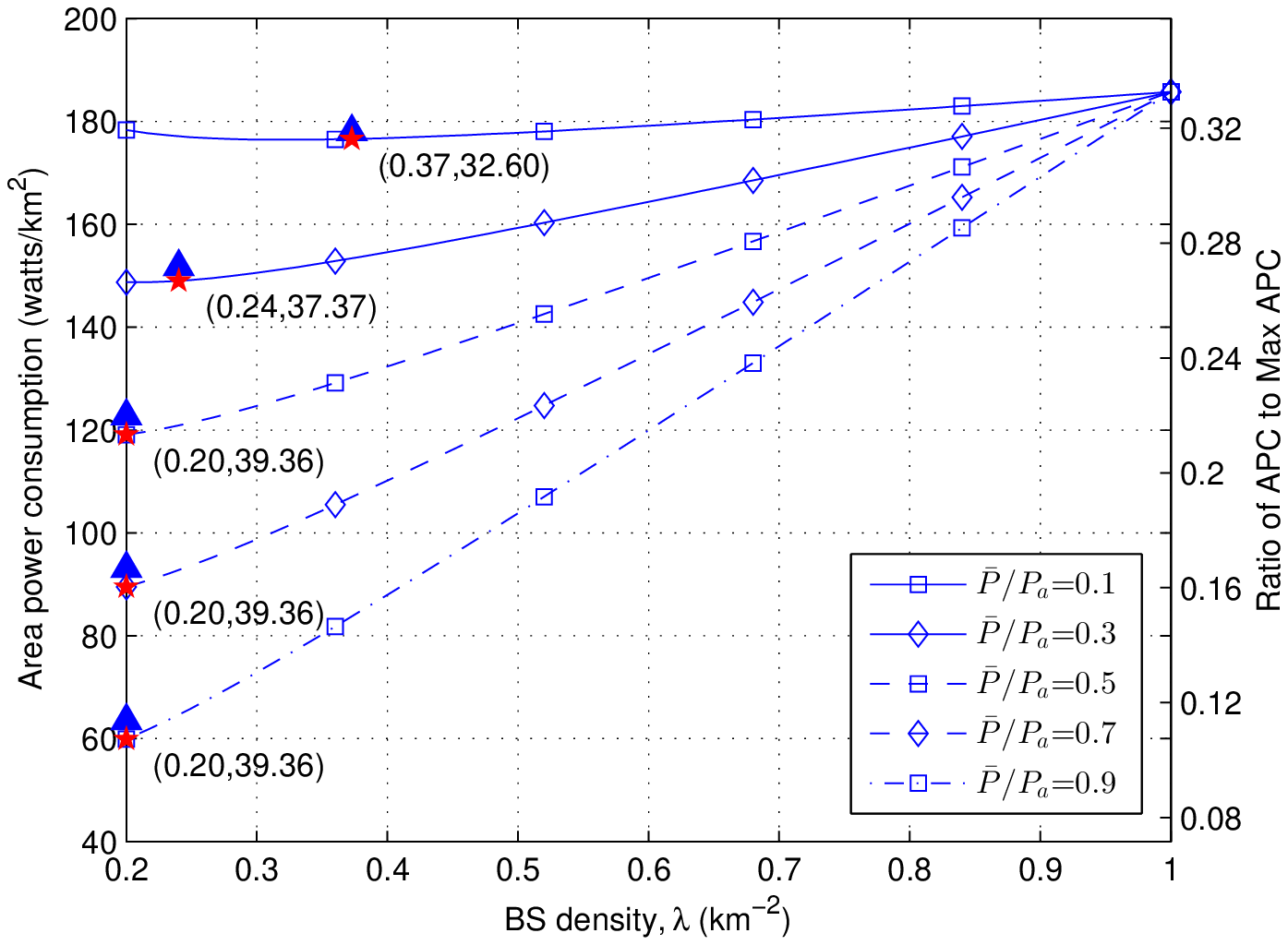}
    \label{fig:ResultApcMinUfrAlpha5}}}
    \caption{APC minimization through BS switching off and transmit power adjustment ($\beta=1$, $\eta=0.8$; solid triangles: $(\hat{\lambda},\hat{p})$, solid stars: $(\hat{\lambda},p^*(\hat{\lambda},1))$; the numbers in parentheses show $(\hat{\lambda},p^*(\hat{\lambda},1))$ in units of km${}^{-2}$ and dBm).}
\label{fig:ResultApcMinUfr}
\vspace{-0.5cm}
\end{figure*}

  Fig.~\ref{fig:ResultApcMinUfr} demonstrates that the design of $\lambda$ and $p$ proposed in Section~\ref{sec:GreenCellFullFreq} functions well in terms of minimizing the APC.
  In the subfigures, the ordinates on the left and right denote the APC values and the ratios of the APC to the maximum APC with $(\lambda_u, P_{max})$, respectively.
  The curves in this figure depict the APC at the optimal transmit power for given $\lambda$ and $\beta=1$ defined in \eqref{eq:Pgiven} and they indicate that this optimization decreases the APC to 15--34\% of that for $(\lambda_u, P_{max})$.
  In addition, these results support that the BS power consumption behavior parameterized by $\bar{P}$ has a dominant impact on the optimal topology adjustment.
  That is, the BS switching off is less effective for small $\bar{P}$, while it is the key to reducing the APC as $\bar{P}$ increases.
  In this figure, $\bar{P}$ is normalized by $P_a$ with the value in Table~\ref{tbl:Parameters}; thus, $0 < \bar{P}/P_a \leq 1$ where $\bar{P}/P_a$ close to zero indicates the quite small effect of the BS sleep mode while $\bar{P}/P_a = 1$ indicates an ideal sleep mode with zero BS power consumption.
  This figure also demonstrates that $(\hat{\lambda}, \hat{p})$ derived in Proposition~\ref{prop:OptPwCnsmTight} provides an excellent suboptimal solution for the joint optimization of $\lambda$ and $p$.
  For example, when $\alpha=4$, $(\hat{\lambda}, \hat{p})$ nearly minimizes the APC and is almost the same as $(\hat{\lambda}, p^*(\hat{\lambda},1))$, which further improves $(\hat{\lambda}, \hat{p})$ as proposed in Section~\ref{subsec:AlgorithmUfr}.
  In contrast, when $\alpha=5$, as a result of the harsh propagation loss, the noise power impact increases relatively and therefore $\psi_e$ has a slight deviation from $\psi$ as shown in Fig.~\ref{fig:ResultContourAlpha5}.
  Thus, through further optimizing $p$ for given $\hat{\lambda}$, the APC performance can be improved.
  Note that in Fig.~\ref{fig:ResultApcMinUfr}, when $\lambda=1$, the curves according to $\bar{P}$ meet at one point because there is no BS in sleep mode.

%
\begin{figure*}[!t]
\centerline{
    \subfigure[Path loss exponent: $\alpha = 4$]{\includegraphics[width=9cm]{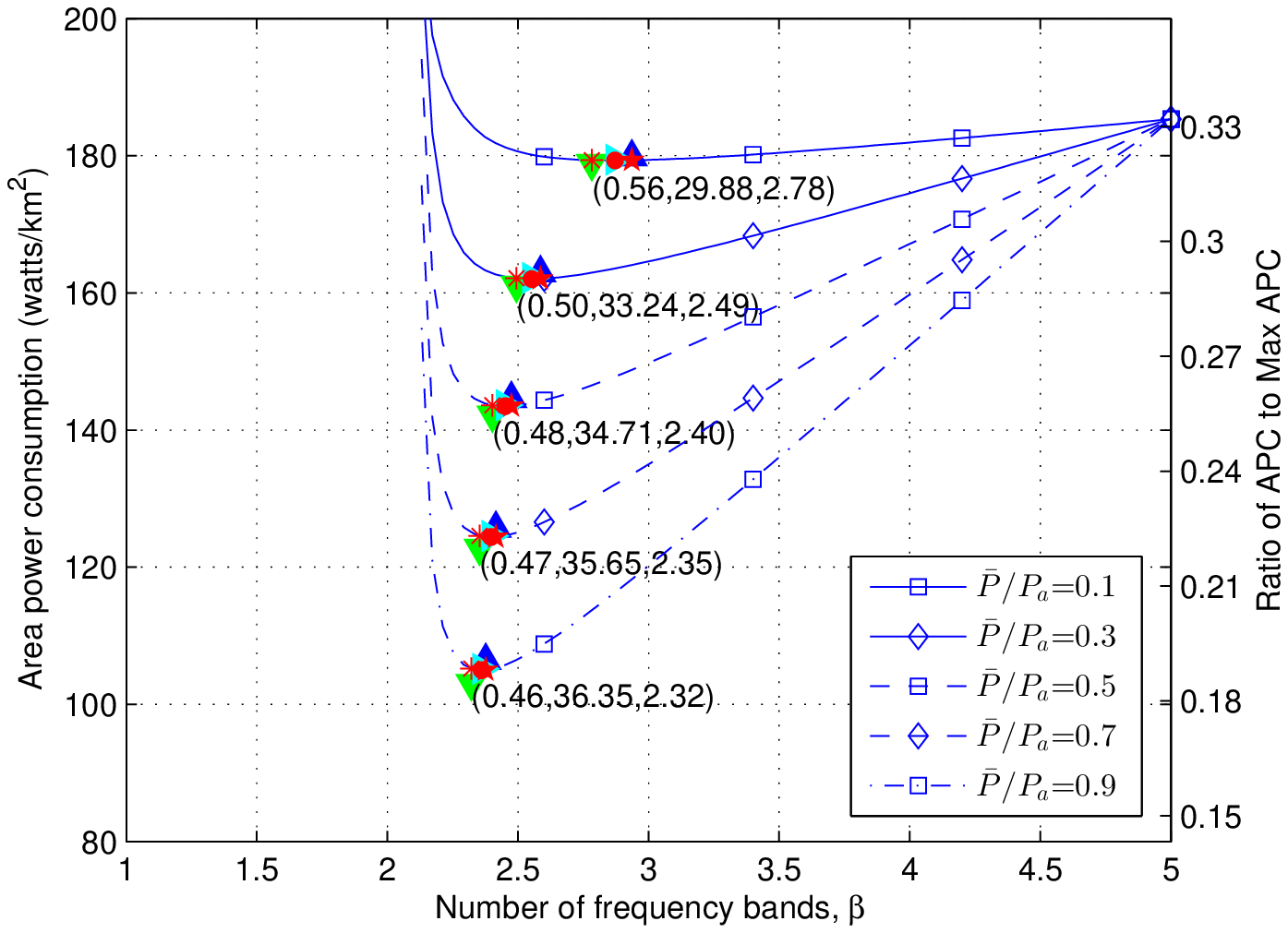}
    \label{fig:ResultApcMinPfrAlpha4}}
\hfil
    \subfigure[Path loss exponent: $\alpha = 5$]{\includegraphics[width=9cm]{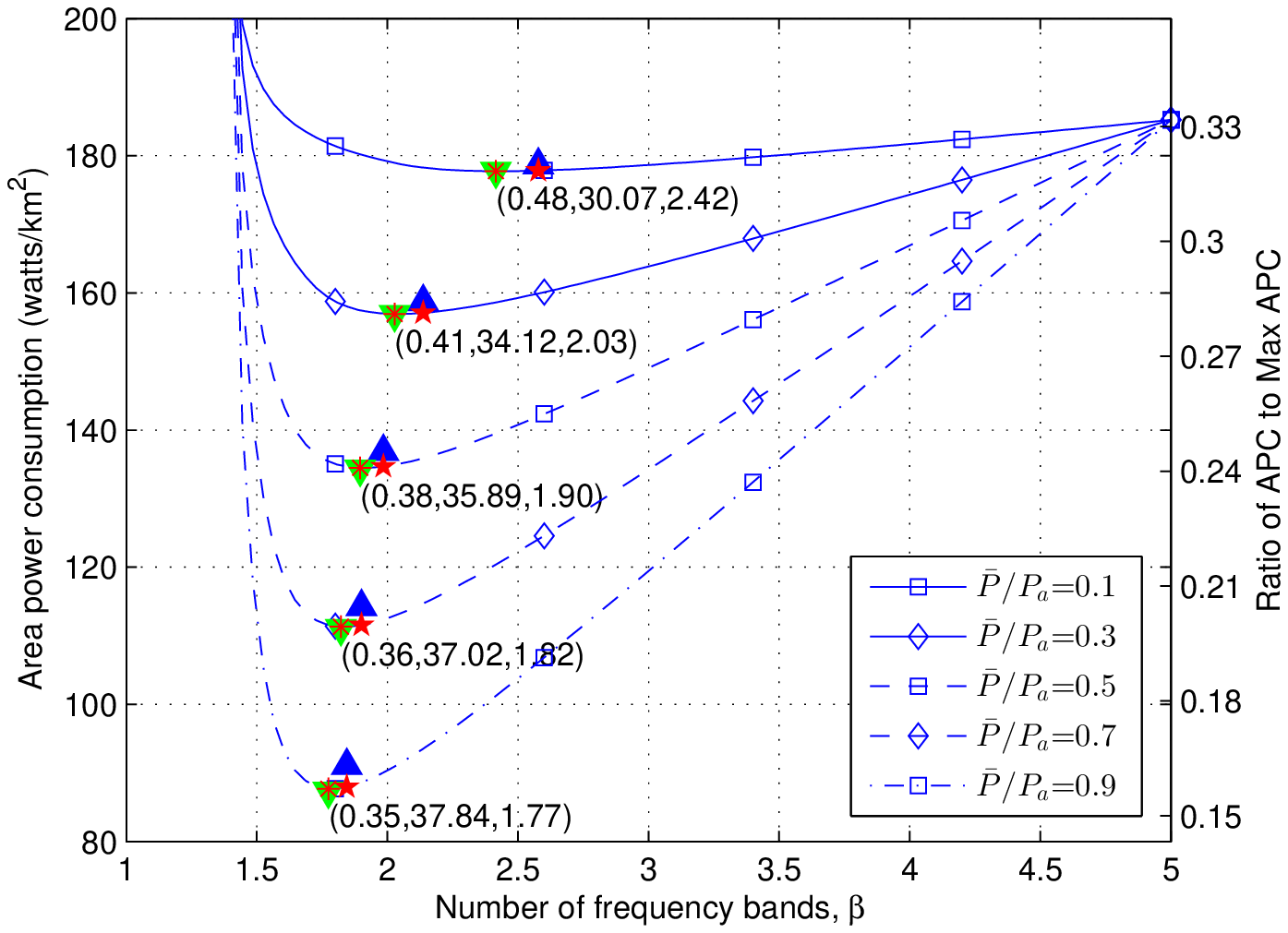}
    \label{fig:ResultApcMinPfrAlpha5}}}
    \caption{The effect of frequency reuse on the APC ($\eta=0.9$; the solid up/down/right triangles denote $(\breve{\lambda},\breve{p},\breve{\beta})$ based on $\psi_\beta$, $\psi_a$, and $\psi_q$, respectively, while the solid star, asterisk, and circle denote $(\breve{\lambda},p^*(\breve{\lambda},\breve{\beta}),\breve{\beta})$ based on $\psi_\beta$, $\psi_a$, and $\psi_q$, respectively; the numbers in parentheses show $(\breve{\lambda},p^*(\breve{\lambda},\breve{\beta}),\breve{\beta})$ based on $\psi_a$ in units of km${}^{-2}$, dBm, and number of frequency bands).}
\label{fig:ResultApcMinPfr}
\vspace{-0.5cm}
\end{figure*}

  Fig.~\ref{fig:ResultApcMinPfr} presents the joint optimization of the number of frequency bands, i.e., $\beta$, together with $\lambda$ and $p$.
  The curves in this figure denote the APC at $(\hat{\lambda},p^{*}(\hat{\lambda},\beta),\beta)$ for given $\beta$ in the abscissa, of which the derivation has been explained in Section~\ref{subsec:DesignPfrGvnBeta}.
  These curves reveal that the optimization of $\beta$ significantly contributes to the minimization of the APC.
  Similar to Fig.~\ref{fig:ResultApcMinUfr}, the curves meet at one point when $\beta=5$, and this is because $\beta\lambda_l \leq \lambda \leq \lambda_u$ from \eqref{eq:MinArPwTightPfr_CnstLambdaL}, \eqref{eq:MinArPwTightPfr_CnstLambdaU}, $\lambda_l = 0.2$, and $\lambda_u = 1$.
  In addition, the results demonstrate that $(\breve{\lambda},\breve{p},\breve{\beta})$ obtained from Proposition~\ref{prop:OptPwCnsmPfr} (in particular, condition (ii)) is an excellent candidate for reducing the APC.
  It is observed that, when $\alpha=4$, $(\breve{\lambda},\breve{p},\breve{\beta})$ for $\psi_a \geq \eta$ may offer an APC less than the minimum APC of a curve, and this implies that $(\breve{\lambda},\breve{p},\breve{\beta})$ for $\psi_a \geq \eta$ may violate constraint $\psi \geq \eta$. 
  In contrast, when $\alpha = 5$, it is demonstrated that $\psi_a$ provides the almost minimum operating point.
  As addressed in Sections~\ref{sec:GreenCellPartialFreq}, the results in this figure support that $(\breve{\lambda},p^*(\breve{\lambda},\breve{\beta}),\breve{\beta})$ for $\psi_\beta$ and $\psi_q$ further improve the APC while that for $\psi_a$ renders the solution feasible.

\begin{figure*}[!t]
\centerline{
    \subfigure[UFR ($\beta=1$, $\eta = 0.8$)]{\includegraphics[width=9cm]{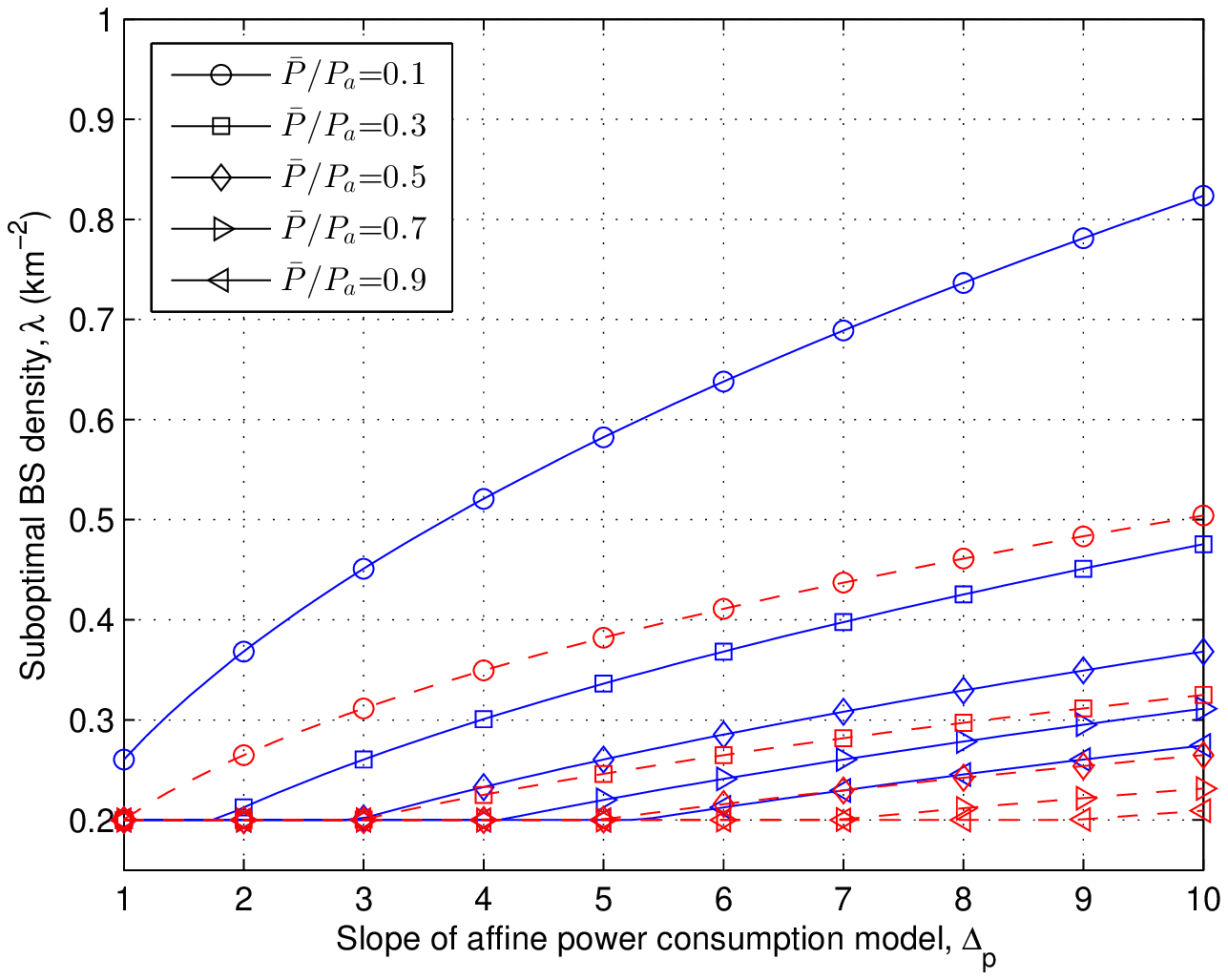}
    \label{fig:ResultPwcModelEffectUfr}}
\hfil
    \subfigure[PFR ($\beta \geq 1$, $\eta = 0.9$)]{\includegraphics[width=9cm]{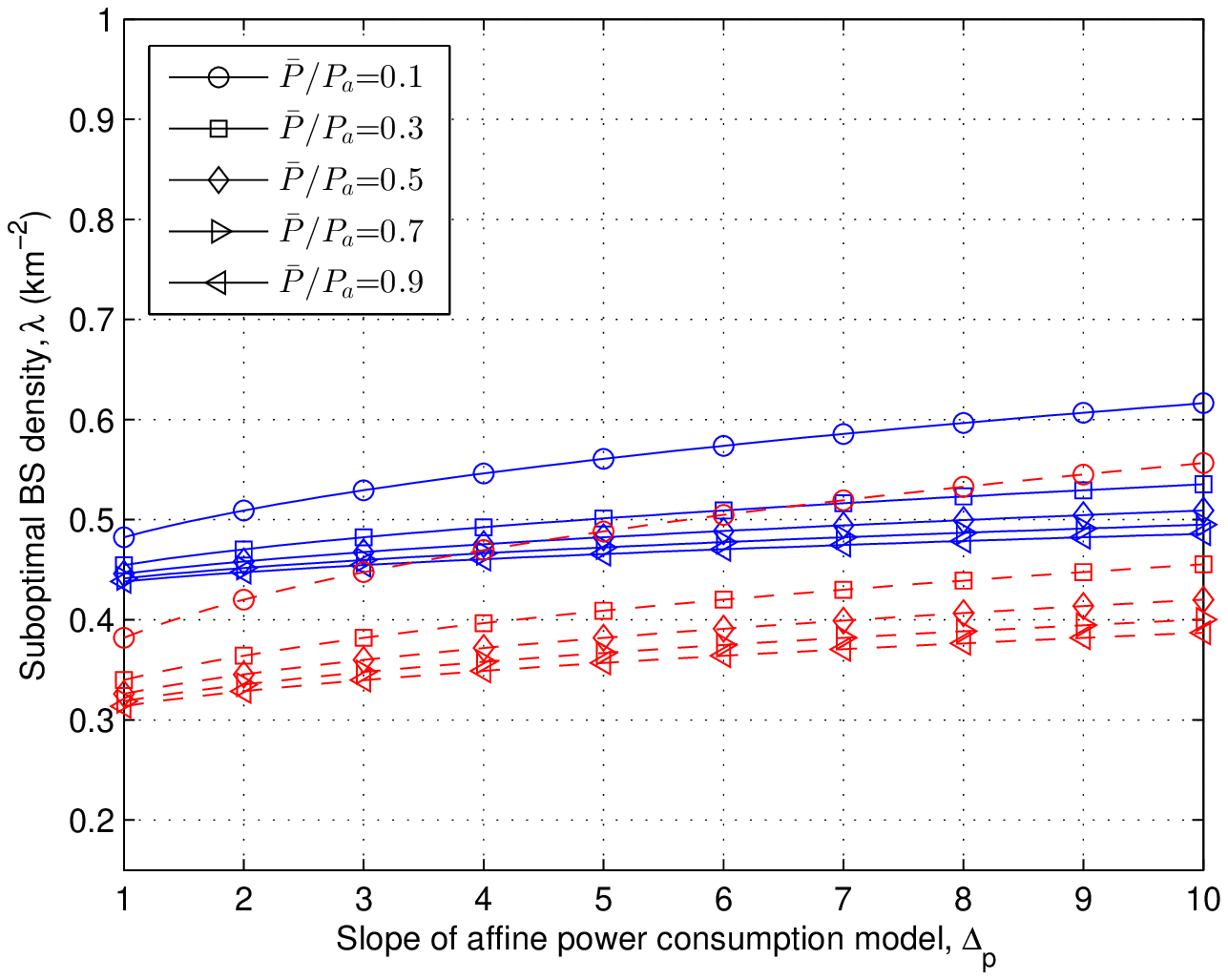}
    \label{fig:ResultPwcModelEffectPfr}}}
    \caption{Impact of affine power consumption model (solid lines: $\alpha = 4$, dashed lines: $\alpha = 5$).}
\label{fig:ResultPwcModelEffect}
\vspace{-0.5cm}
\end{figure*}

  Fig.~\ref{fig:ResultPwcModelEffect} elaborates the impact of the affine BS power consumption model in \eqref{eq:PsModel} on the BS switching off.
  The ordinates of the two subfigures denote $\hat{\lambda}$ in Proposition~\ref{prop:OptPwCnsmTight} for the UFR and $\breve{\lambda}$ derived from the conditions in Proposition~\ref{prop:OptPwCnsmPfr} for the PFR, respectively. 
  As already observed in Figs.~\ref{fig:ResultApcMinUfr} and \ref{fig:ResultApcMinPfr}, $\hat{\lambda}$ and $\breve{\lambda}$ decrease as $\bar{P}$ increases.
  Because $\bar{P}$ signifies the power saving effect in the sleep mode, BS switching off increasingly revs up as $\bar{P}$ increases.
  On the other hand, it is observed that the density of the active BSs also increases with $\Delta_p$.
  Note that both $\lambda$ and $p$ increase the APC in \eqref{eq:DefAreaPwCons} by a factor of $\Delta_p$ and the coverage probability constraint expressed in terms of $\lambda^{\frac{\alpha}{2}} p$ is more sensitive to $\lambda$ compared with $p$ for given $\beta$.
  Therefore, in the APC minimization, as $\Delta_p$ increases, $\lambda$ increases while $p$ decreases.
  It is observed that this feature remains consistent regardless of the propagation loss model denoted by $\alpha$.

\section{Conclusions} \label{sec:Conclusions}

  This paper investigated the minimization of the area power consumption of BSs under multicell coverage and capacity constraints.
  The design problems were expressed as GP and their optimal solution and optimality conditions revealed that the operation for spatial topology adjustment, e.g., only reducing the transmit power, only switching off the BSs, and both switching off the BS and adjusting the transmit power, is determined based on the amount of power saving that results from the BS switching off.
  Furthermore, this operation depends on the BS power consumption behaviors, wireless environments, and target network performances.
  Even though simple models for mathematical tractability, e.g., a simple network capacity constraint and homogeneous PPP, were assumed, the network-wide results in this paper can be used as design guidelines of the BS switching off operations for green multicell networks.
  As future work, it would be interesting to investigate the design problem for greening the entire access network through considering the power consumption of UEs as well as BSs.


\appendices

\section{Proof of the increase in $\psi(\lambda,p,1)$ with $\lambda$ and $p$ in Remark~\ref{rmk:Increasing}} \label{app:proof:rmk:Increasing}


  It is shown that $\frac{\partial \psi(\lambda,p,1)}{\partial \lambda} > 0$ and $\frac{\partial \psi(\lambda,p,1)}{\partial p} > 0$.
  Let $a \triangleq \pi(1+\phi(\xi,\alpha)) > 0$ and $b \triangleq \xi \sigma^2 > 0$.
  Then, $\psi(\lambda, p, 1)$ in \eqref{eq:CoverPrPsi} is expressed as $\pi \lambda \int_0^\infty \exp\left( -a\lambda x - b p^{-1} x^{\frac{\alpha}{2}} \right) dx$.
  Accordingly,
\begin{align}\begin{aligned}\label{eq:PartialPhiLambda}
  \frac{\partial \psi(\lambda,p,1)}{\partial \lambda} & = \pi \int_0^{\infty} \left(  1 - a \lambda x \right) \exp\left(- a \lambda x - b p^{-1} x^{\frac{\alpha}{2}} \right) dx \\
    & \stackrel{\mathrm{(a)}}{\geq} \pi \int_0^\infty \left( 1 - a\lambda x \right) \left( 1 - b p^{-1} x^{\frac{\alpha}{2}} \right) \exp\left( - a \lambda x \right) dx \\
    & = \frac{\pi \alpha}{2} \left( a \lambda \right)^{-\frac{\alpha}{2}-1} b p^{-1} \Gamma\left( 1 + \frac{\alpha}{2} \right) > 0,
\end{aligned}\end{align}
  where (a) follows from $\exp(-x) \geq 1 - x$ for $x \geq 0$\footnote{This is shown in Section~\ref{subsec:DesignUfr}.}.
  On the other hand,
\begin{align}\begin{aligned}\label{eq:PartialPhiP}
  \frac{\partial \psi(\lambda,p,1)}{\partial p} = \pi b \lambda p^{-2} \int_0^{\infty} x^{\frac{\alpha}{2}} \exp\left( -a\lambda x - b p^{-1} x^{\frac{\alpha}{2}} \right) dx,
\end{aligned}\end{align}
  where the integrand is positive except when $x=0$.
  Hence, $\frac{\partial \psi(\lambda,p,1)}{\partial p} > 0$.
  \hfill\IEEEQED

\section{Proof of Proposition~\ref{prop:OptPwCnsmTight}} \label{app:proof:prop:OptPwCnsmTight}

  Let a variable with superscript $(g)$ denote the logarithm of an original variable.
  That is, $x^{(g)}$ represents $\log x$ for $x>0$.
  Then, when $\kappa \triangleq \frac{\alpha}{2}$, problem \eqref{eq:MinArPwTight} is equivalent to the following GP:
\begin{subequations}\label{eq:MinArPwTightGeom}
\begin{align}
    \underset{\lambda^{(g)},\,p^{(g)}}{\text{minimize}} \quad   & \log \left( \exp(\lambda^{(g)}+\bar{P}^{(g)}) + \exp(\lambda^{(g)}+p^{(g)}+\Delta_p^{(g)}) \right) \label{eq:MinArPwTightGeom_Obj} \\
    \text{subject to} \quad & -\kappa\lambda^{(g)} - p^{(g)} + \theta^{(g)} \leq 0 \label{eq:MinArPwTightGeom_CnstCoverage} \\
    & -\lambda^{(g)} + \lambda^{(g)}_l \leq 0 \label{eq:MinArPwTightGeom_CnstLambdaL} \\
    & \lambda^{(g)} - \lambda^{(g)}_u \leq 0 \label{eq:MinArPwTightGeom_CnstLambdaU} \\
    & p^{(g)} - P^{(g)}_{max} \leq 0 \label{eq:MinArPwTightGeom_CnstTxPw}.
\end{align}
\end{subequations}
%
  Note that problem \eqref{eq:MinArPwTightGeom} is a convex optimization problem, because a log-sum-exponential function is convex and the feasible set is the intersection of the sublevel set of affine functions \cite{Boyd07_CvxOpt}.
  The assumption of $\lambda_u^2 P_{max} > \theta(\xi,\eta)$ guarantees that there exists $(\lambda^{(g)},p)$ to meet affine inequalities \eqref{eq:MinArPwTightGeom_CnstCoverage}--\eqref{eq:MinArPwTightGeom_CnstTxPw}.
  That is, the Slater's condition holds.
  This implies that the KKT conditions provide the necessary and sufficient conditions for optimality in problem \eqref{eq:MinArPwTightGeom}.
  Therefore, from the KKT conditions of problem \eqref{eq:MinArPwTightGeom}, $\hat{\lambda}^{(g)}$ and $\hat{p}^{(g)}$ can be expressed as functions of $\kappa$, $\bar{P}^{(g)}$, $\Delta_p^{(g)}$, $\lambda^{(g)}_l$, $\lambda^{(g)}_u$, $P^{(g)}_{max}$, and $\theta^{(g)}$ according to the intervals of $\bar{P}^{(g)}$ and this solution can be rewritten in terms of original environmental variables, like \eqref{eq:OptSolTight}.
\hfill\IEEEQED

\section{Proof of Proposition~\ref{prop:OptPwCnsmPfr}} \label{app:proof:prop:OptPwCnsmPfr}

  Like Appendix~\ref{app:proof:prop:OptPwCnsmTight}, for $x>0$, $x^{(g)}$ denotes $\log x$.
  Then, problem \eqref{eq:MinArPwTightPfrGeom} can be recast into
\begin{subequations}\label{eq:MinArPwGenTightFrGeom}
\begin{align}
  \underset{\lambda^{(g)},\,p^{(g)},\,\beta^{(g)}}{\text{minimize}} \quad   & \log\left( \exp\left( \lambda^{(g)} + \bar{P}^{(g)}  \right) + \exp\left( \Delta_p^{(g)} + \lambda^{(g)} + p^{(g)}  \right) \right) \label{eq:MinArPwGenTightFrGeom_Obj} \\
  \text{subject to} \quad & \log\left( \exp\left( c^{(g)}_0 - \kappa\lambda^{(g)} - p^{(g)} - \beta^{(g)} \right) + \exp\left( c^{(g)}_1 - \beta^{(g)} \right) + \exp\left( c^{(g)}_2 - 2\beta^{(g)} \right) \right) \leq 0 \label{eq:MinArPwGenTightFrGeom_CnstCoverage} \\
  & \lambda^{(g)}_l - \lambda^{(g)} + \beta^{(g)} \leq 0 \label{eq:MinArPwGenTightFrGeom_CnstLambdaL} \\
  & -\lambda^{(g)}_u + \lambda^{(g)} \leq 0 \label{eq:MinArPwGenTightFrGeom_CnstLambdaU} \\
  & p^{(g)} - P^{(g)}_{max} \leq 0 \label{eq:MinArPwGenTightFrGeom_CnstTxPw} \\
  & -\beta^{(g)} \leq 0 \label{eq:MinArPwGenTightFrGeom_CnstDelta}.
\end{align}
\end{subequations}
  Note that problem \eqref{eq:MinArPwGenTightFrGeom} is a convex problem.
  The Lagrangian associated with problem \eqref{eq:MinArPwGenTightFrGeom} is given as follows:
\begin{align}\begin{aligned}\label{eq:LaplacianFr}
  & L_\beta(\lambda^{(g)},\,p^{(g)},\,\beta^{(g)},\,\tau_c,\,\tau_{\lambda l},\,\tau_{\lambda u},\,\tau_p,\,\tau_\beta) \triangleq \log\left( \exp\left( \lambda^{(g)} + \bar{P}^{(g)}  \right) + \exp\left( \Delta_p^{(g)} + \lambda^{(g)} + p^{(g)}  \right) \right) \\
    & \hspace{1cm} + \tau_c \log\left( \exp\left( c^{(g)}_0 - \kappa\lambda^{(g)} - p^{(g)} - \beta^{(g)} \right) + \exp\left( c^{(g)}_1 - \beta^{(g)} \right) + \exp\left( c^{(g)}_2 - 2\beta^{(g)} \right) \right) \\
    & \hspace{1cm} + \tau_{\lambda l} \left( \lambda^{(g)}_l - \lambda^{(g)} + \beta^{(g)} \right) + \tau_{\lambda u} \left( -\lambda^{(g)}_u + \lambda^{(g)} \right) + \tau_p \left( p^{(g)} - P^{(g)}_{max} \right) - \tau_\beta \beta^{(g)}.
\end{aligned}\end{align}
  where $\tau_c$, $\tau_{\lambda l}$, $\tau_{\lambda u}$, $\tau_{p}$ and $\tau_{\beta}$ are the nonnegative dual variables associated with constraints \eqref{eq:MinArPwGenTightFrGeom_CnstCoverage}--\eqref{eq:MinArPwGenTightFrGeom_CnstDelta}. 
  The KKT conditions of problem \eqref{eq:MinArPwGenTightFrGeom} include $\nabla L_\beta = \mathbf{0}$, which is equivalent to
\begin{align}
  & \frac{\partial L_\beta}{\partial \lambda^{(g)}} = 1 - \breve{\tau}_c \frac{\kappa \exp \left( c^{(g)}_0 - \kappa\breve{\lambda}^{(g)} - \breve{p}^{(g)} - \breve{\beta}^{(g)} \right)}{\exp\left( c^{(g)}_0 - \kappa\breve{\lambda}^{(g)} - \breve{p}^{(g)} - \breve{\beta}^{(g)} \right) + \exp\left( c^{(g)}_1 - \breve{\beta}^{(g)} \right) + \exp\left( c^{(g)}_2 - 2\breve{\beta}^{(g)} \right)} \nonumber \\
    & \hspace{2cm} - \breve{\tau}_{\lambda l} + \breve{\tau}_{\lambda u} = 0, \label{eq:Vanish_LaplacianFrLambda} \\
  & \frac{\partial L_\beta}{\partial p^{(g)}} = \frac{\exp\left( \Delta_p^{(g)} + \breve{p}^{(g)} \right)}{\exp\left( \bar{P}^{(g)} \right) + \exp\left( \Delta_p^{(g)} + \breve{p}^{(g)} \right)} \nonumber \\
    & \hspace{2cm} - \breve{\tau}_c \frac{\exp \left( c^{(g)}_0 - \kappa\breve{\lambda}^{(g)} - \breve{p}^{(g)} - \breve{\beta}^{(g)} \right)}{\exp \left( c^{(g)}_0 - \kappa\breve{\lambda}^{(g)} - \breve{p}^{(g)} - \breve{\beta}^{(g)} \right) + \exp\left( c^{(g)}_1 - \breve{\beta}^{(g)} \right) + \exp\left( c^{(g)}_2 - 2\breve{\beta}^{(g)} \right)} \nonumber \\
    & \hspace{2cm} + \breve{\tau}_p = 0, \label{eq:Vanish_LaplacianFrP} \\
  & \frac{\partial L_\beta}{\partial \beta^{(g)}} = - \breve{\tau}_c \frac{\exp \left( c^{(g)}_0 - \kappa\breve{\lambda}^{(g)} - \breve{p}^{(g)} - \breve{\beta}^{(g)} \right) + \exp\left( c^{(g)}_1 - \breve{\beta}^{(g)} \right) + 2\exp\left( c^{(g)}_2 - 2\breve{\beta}^{(g)} \right)}{\exp \left( c^{(g)}_0 - \kappa\breve{\lambda}^{(g)} - \breve{p}^{(g)} - \breve{\beta}^{(g)} \right) + \exp\left( c^{(g)}_1 - \breve{\beta}^{(g)} \right) + \exp\left( c^{(g)}_2 - 2\breve{\beta}^{(g)} \right)} \nonumber \\
    & \hspace{2cm} + \breve{\tau}_{\lambda l} - \breve{\tau}_{\beta} = 0, \label{eq:Vanish_LaplacianFrDelta}
\end{align}
  where $\breve{\tau}_c$, $\breve{\tau}_{\lambda l}$, $\breve{\tau}_{\lambda u}$, $\breve{\tau}_{p}$ and $\breve{\tau}_{\beta}$ denote the optimal dual variables of problem \eqref{eq:MinArPwGenTightFrGeom}.
  It is worth noting that $\breve{\tau}_c > 0$, which follows from \eqref{eq:Vanish_LaplacianFrP}.
  This means that the associated constraint is active, from the complementary slackness, which is one of the KKT conditions.
  That is,
\begin{align} \label{eq:OptCndtFrCoverage}
  \exp\left( c^{(g)}_0 - \kappa\breve{\lambda}^{(g)} - \breve{p}^{(g)} - \breve{\beta}^{(g)} \right) + \exp\left( c^{(g)}_1 - \breve{\beta}^{(g)} \right) + \exp\left( c^{(g)}_2 - 2\breve{\beta}^{(g)} \right) = 1.
\end{align}
  From $\breve{\tau}_c > 0$ and \eqref{eq:Vanish_LaplacianFrDelta}, $\breve{\tau}_{\lambda l} > 0$; hence, similar to \eqref{eq:OptCndtFrCoverage}, from the complementary slackness,
\begin{align} \label{eq:OptCndtFrDeltaLambda}
  \breve{\lambda}^{(g)} = \breve{\lambda}^{(g)}_l + \breve{\beta}^{(g)}.
\end{align}
  Therefore, \eqref{eq:OptPwCnsmPfr} is yielded from \eqref{eq:OptCndtFrCoverage} and \eqref{eq:OptCndtFrDeltaLambda}.

  Cases (i), (iii), and (iv) represent the extreme instances where one of constraints \eqref{eq:MinArPwGenTightFrGeom_CnstLambdaU}, \eqref{eq:MinArPwGenTightFrGeom_CnstTxPw}, and \eqref{eq:MinArPwGenTightFrGeom_CnstDelta} is active.
  In contrast, case (ii) signifies that inequality constraints \eqref{eq:MinArPwGenTightFrGeom_CnstLambdaU}, \eqref{eq:MinArPwGenTightFrGeom_CnstTxPw}, and \eqref{eq:MinArPwGenTightFrGeom_CnstDelta} are inactive; thus, $\tau_{\lambda u} = 0$, $\tau_p = 0$, and $\tau_\beta = 0$.
  From \eqref{eq:Vanish_LaplacianFrLambda}--\eqref{eq:OptCndtFrDeltaLambda} and $(\tau_{\lambda u},\tau_p,\tau_\beta) = (0,0,0)$, the result of case (ii) is derived.
%
%
\hfill\IEEEQED


\bibliographystyle{IEEEtran}
\bibliographystyle{IEEEbib}
\bibliography{./tskIEEEabrv,./tskRef}


\end{document}